\documentclass{llncs}
\pdfoutput=1

\newif\iffullversion
\fullversiontrue
\fullversiontrue

\usepackage{etoolbox}
\AtEndEnvironment{proof}{\qed}

\usepackage{hyperref}
\hypersetup{
    colorlinks=true, %
    linkcolor=blue,  %
    citecolor=blue,  %
    urlcolor=blue    %
}

\usepackage[utf8]{inputenc}
\usepackage[T1]{fontenc}
\usepackage{amsmath,amsfonts,amssymb,mathtools}
\usepackage{graphicx}
\usepackage{subcaption}

\renewcommand{\figurename}{Figure}
\usepackage{tikz-cd}
\usetikzlibrary{calc}
\usepackage{booktabs}

\newcommand{\Match}{\mu}
\newcommand{\Proj}[2]{{%
  \left.\kern-\nulldelimiterspace%
  #1%
  \vphantom{\big|}%
  \right|_{#2}%
}}
\newcommand{\Block}[2]{\operatorname{block}(#1, #2)}
\newcommand{\Res}[2]{\operatorname{res}(#1, #2)} %
\newcommand{\SAT}{(2,2)\mbox{-E3-SAT}}
\newcommand{\Prop}{\phi}
\newcommand{\HR}{\mbox{HR}}
\newcommand{\HRS}{\mbox{HRS}}
\newcommand{\HRSP}{\mbox{HRS-P}}
\newcommand{\HRSPS}{\mbox{HRS-PS}}
\newcommand{\HRSFC}{\mbox{HRS-FC}}
\newcommand{\HRSC}{\mbox{HRS-C}}
\newcommand{\HRIC}{\mbox{HRIC}}

\newcommand{\Size}[1]{\operatorname{size}(#1)}
\renewcommand{\Cap}[1]{\operatorname{cap}(#1)}
\newcommand{\Hosps}[1]{\operatorname{hosps}(#1)}
\newcommand{\Bdry}[3]{\operatorname{bdry}(#1, #2, #3)}
\newcommand{\BijFn}{\varphi}
\newcommand{\Bij}[2]{\BijFn(#1, #2)}

\usepackage{color}

\makeatletter
\@addtoreset{lemma}{section}          %
\makeatother

\iffullversion
\newenvironment{deferproof}[1]{\begin{proof}}{\end{proof}}
\else
\usepackage{environ}
\newcommand{\deferredProofContent}{}
\newtoks\DeferProofBodyToks
\NewEnviron{deferproof}[1]{%
  \DeferProofBodyToks=\expandafter{\BODY}%
  \xdef\deferredProofContent{%
    \unexpanded\expandafter{\deferredProofContent}%
    \unexpanded{\begingroup\def\proofname{Proof of #1}\begin{proof}}%
    \unexpanded\expandafter{\the\DeferProofBodyToks}%
    \unexpanded{\end{proof}\endgroup}%
  }%
}

\fi

\begin{document}

\title{Hospitals/Residents with Inseparable Couples: Finding a
  Coalition-Stable Assignment Is NP-Hard}

\author{Zeyuan Hu\inst{1} \and
  C. Gregory Plaxton \inst{2}
}
\institute{University of California, Los Angeles
\email{(zeyuanhu@cs.ucla.edu)}\\ \and
University of Texas at Austin\\
\email{(plaxton@cs.utexas.edu)}}
\maketitle
\begin{abstract}
In recent work on course allocation, Rodr\'{i}guez and Manlove
consider the complexity of finding a stable assignment under four
notions of stability, including two coalitional notions. In one case,
which they call pair-size stability, they show that a stable
assignment always exists and they provide a polynomial-time algorithm
to find one. In a second case, called pair stability, they observe
that an earlier NP-hardness result of McDermid and Manlove holds for a
special case of course allocation called \emph{Hospitals/Residents
with Sizes} ($\HRS$). In a third case, called first-coalition
stability, they use a reduction from $\HRS$ to show it is NP-hard to
find a stable assignment. They leave open the complexity of finding a
stable assignment under so-called coalition stability. Building on
ideas from McDermid and Manlove, we resolve the open problem of
Rodr\'{i}guez and Manlove by showing that it is NP-hard to find a
coalition-stable assignment for $\HRS$. Indeed, our proof shows that
the problem remains NP-hard when the hospital capacities and resident
sizes are at most two. Accordingly, our NP-hardness result applies to
the special case of $\HRS$ known as \emph{Hospitals/Residents with
Inseparable Couples} ($\HRIC$). Finally, we introduce a novel and
natural notion of coalitional stability for both $\HRS$ and course
allocation, and we show that our NP-hardness result extends to this
notion, which we call unitwise-coalition stability.

\end{abstract}

\section{Introduction}
\label{sec:intro}

We begin by defining a well-known ``hospitals and residents'' variant
of the classic stable marriage problem~\cite{Gale1962}, which we
denote $\HR$.\footnote{In~\cite{Gale1962}, this is referred to as the
college admissions problem.} We are given a set of hospitals and a set
of residents; the goal is to determine a ``good'' assignment of residents
to hospitals. Each resident (resp., hospital) has strict preferences
over the set of hospitals (resp., residents) that it finds
acceptable.\footnote{Since we will only consider assigning resident
$r$ to hospital $h$ if $h$ is acceptable to $r$ and $r$ is acceptable
to $h$, we can assume without loss of generality that $h$ appears on
the preference list of $r$ if and only if $r$ appears on the
preference list of $h$.} Each hospital $h$ has a positive integer
capacity that represents the maximum number of residents that can be
assigned to $h$. An assignment for a given $\HR$ instance is a set
$\Match$ of resident-hospital pairs (if $(r,h)$ belongs to $\Match$,
then $\Match$ is said to assign resident $r$ to hospital $h$) subject
to the following constraints.
\begin{enumerate}
  \item Acceptability constraints. If $(r,h)$ belongs to $\Match$, then $h$ is
    acceptable to $r$ and $r$ is acceptable to $h$.
  \item Resident degree constraints. Any resident appears in at most
    one pair in $\Match$.
  \item Hospital capacity constraints. The number of pairs in $\Match$
    involving a hospital $h$ is at most the capacity of $h$.
\end{enumerate}
A resident $r$ is said to be unassigned under $\Match$ if no pair in
$\Match$ involves $r$. For any assignment $\Match$ and any hospital
$h$, we define $\Res{\Match}{h}$ as $\{r\mid (r,h) \in \Match\}$. A
hospital $h$ is said to be under-subscribed under $\Match$ if
$|\Res{\Match}{h}|$ is less than the capacity of $h$. A
resident-hospital pair $(r,h)$ is said to be a blocking pair for an
assignment $\Match$ if (1) $r$ is unassigned under $\Match$ or $r$
prefers $h$ to its assigned hospital under $\Match$ and (2) $h$ is
under-subscribed under $\Match$ or $h$ prefers $r$ to its least
preferred assigned resident under $\Match$. An assignment is said to
be stable if it does not admit a blocking pair.

In this paper we are concerned with a generalization of the $\HR$
problem called \emph{Hospitals/Residents with Sizes} ($\HRS$)
introduced by McDermid and Manlove~\cite{McDermid2010}. In this
generalization, each resident has a positive integer size and the
hospital capacity constraints are modified to require that the total
size of the residents assigned to a hospital $h$ is at most the
capacity of $h$. McDermid and Manlove consider a resident-hospital
pair $(r,h)$ to be a blocking pair for an assignment $\Match$ if (1)
$r$ is unassigned under $\Match$ or $r$ prefers $h$ to its assigned
hospital under $\Match$ and (2) the size of $\{r' \in
\Res{\Match}{h}\mid \mbox{$h$ prefers $r$ to $r'$}\}$ is at least the
size of $r$ minus the unused capacity of $h$ under $\Match$. As for
the $\HR$ problem, an assignment is said to be stable if it does not
admit a blocking pair.

Alternative notions of stability have been proposed for
$\HRS$. Adopting the terminology of Rodr\'{i}guez and
Manlove~\cite{Rodriguez2025}, we refer to the stability notion
investigated by McDermid and Manlove~\cite{McDermid2010} as pair
stability. In the concluding remarks of their paper, McDermid and
Manlove propose a second notion of stability for $\HRS$ in which
condition (2) above is replaced by ``there exists a subset of $\{r'
\in \Res{\Match}{h}\mid \mbox{$h$ prefers $r$ to $r'$}\}$ with total
size at most the size of $r$ and at least the size of $r$ minus the
unused capacity of $h$ under $\Match$.''  Rodr\'{i}guez and Manlove
refer to a resident-hospital pair $(r,h)$ satisfying (1) and the
modified version of (2) as a size-blocking pair. The associated notion
of stability (i.e., absence of size-blocking pair) has subsequently
been referred to as pair-size stability~\cite{Rodriguez2025},
occupancy-stability~\cite{Balasundaram2025}, and
size-stability~\cite{Delacretaz2019}.\footnote{Delacr\'{e}taz studies
the special case where the resident sizes belong to $\{1,2\}$.} In the
present paper, we use the term pair-size stability.

Rodr\'{i}guez and Manlove~\cite{Rodriguez2025} study a generalization
of $\HRS$ to the setting of course allocation, where the goal is to
assign students to courses. Unlike $\HRS$, where each resident can be
assigned to at most one hospital, a student can take many courses and
a course can be taken by many students.\footnote{We can model an
$\HRS$ instance $I$ as a course allocation instance $I'$ as follows:
for each size-$k$ resident in $I$, we add a unit-capacity $k$-credit
course to $I'$; for each capacity-$k$ hospital in $I$, we add a
student to $I'$ with credit limit $k$.} Rodr\'{i}guez and Manlove
prove that a pair-size stable assignment always exists in the course
allocation setting, and they provide a polynomial-time algorithm for
finding such an assignment. Independently, Balasundaram et
al.~\cite{Balasundaram2025} obtain the same result for pair-size
stability (which they call occupancy stability) in the $\HRS$ setting.

Rodr\'{i}guez and Manlove propose and investigate two ``coalitional''
notions of stability for course allocation, which they call
first-coalition stability and coalition stability. For these notions a
blocking coalition is given by a set of courses and a student, as
opposed to a course-student pair. These coalitional notions are also
applicable to the $\HRS$ special case of course allocation. Using
notation similar to that introduced by Rodr\'{i}guez and Manlove, we
define $\HRSC$ (resp., $\HRSP$, $\HRSPS$, and $\HRSFC$) as the problem
of finding a coalition-stable (resp., pair-stable, pair-size-stable,
first-coalition-stable) assignment for a given $\HRS$ instance, or
reporting correctly that no such assignment exists. McDermid and
Manlove show that $\HRSP$ is NP-hard~\cite{McDermid2010}. As discussed
earlier, Rodr\'{i}guez and Manlove provide a polynomial-time
algorithm for finding a pair-size stable assignment in the course
allocation setting (which generalizes $\HRSPS$). Rodr\'{i}guez and
Manlove show that $\HRSFC$ is NP-hard, and leave open the complexity
of $\HRSC$.

We now define the two coalitional notions of stability introduced by
Rodr\'{i}guez and Manlove~\cite{Rodriguez2025}. Given the NP-hardness
results discussed in the previous paragraph, we find it sufficient to
define these notions in the restricted setting of $\HRS$, as opposed
to the more general course allocation setting. Given an assignment
$\Match$ for an $\HRS$ instance, Rodr\'{i}guez and Manlove define a
blocking (resp., first-blocking) coalition for $\Match$ as a pair
$(R,h)$ where $R$ is a set of residents and $h$ is a hospital such
that (1) for all $r$ in $R$, $r$ is unassigned under $\Match$ or $r$
prefers $h$ to its assigned hospital under $\Match$, and (2) there
exists a subset of $\{r' \in \Res{\Match}{h}\mid \mbox{$h$ prefers $r$
  to $r'$ for all (resp., some) $r$ in $R$}\}$ with total size at most
the total size of $R$ and at least the total size of $R$ minus the
unused capacity of $h$ under $\Match$.\footnote{Our definitions of the
terms ``blocking coalition'' and ``first-blocking coalition'', as well
as our earlier definitions of the terms ``blocking pair'' and
``size-blocking pair'', are formulated somewhat differently than in
earlier work, but are equivalent to the original definitions.}

Rodr\'{i}guez and Manlove~\cite{Rodriguez2025} consider six different
problems (e.g., testing whether a given assignment is stable, finding
a stable assignment) with respect to each of the four notions of
stability discussed above. For some of the resulting 24 problem
combinations, there is both an existence question (i.e., whether a
stable assignment is guaranteed to exist) and a complexity question
(i.e., what is the complexity of finding a stable assignment if one
exists) to resolve. Rodr\'{i}guez and Manlove resolve the questions
associated with all 24 problem combinations, except they leave open
the question of whether a coalition-stable assignment is guaranteed to
exist for any course allocation instance, along with the associated
complexity question.\footnote{Regarding the problem of testing whether
  a given assignment is coalition-stable, Rodr\'{i}guez and Manlove
  establish co-NP-completeness in the context of course allocation;
  their proof also applies to the $\HRS$ special case.} In the present
paper, we resolve these questions by exhibiting an $\HRS$ instance
with four residents and three hospitals that does not admit a
coalition-stable assignment, and by proving that $\HRSC$ is
NP-hard.\footnote{In~\cite{Rodriguez2025}, Rodr\'{i}guez and Manlove
  also consider problem variants with so-called ``downward-feasible
  constraints''. Our hardness results hold even in the absence of such
  constraints.}

We show that $\HRSC$ is NP-hard using the same high-level approach as
McDermid and Manlove~\cite[Theorem 3.8, erratum]{McDermid2010} used to
show that $\HRSP$ is NP-hard. More specifically, we also use a
reduction from a special case of 3-SAT called $\SAT$; this
satisfiability variant is formally defined in
Section~\ref{sec:prelim}. The transformation of~\cite{McDermid2010}
from $\SAT$ to $\HRSP$ is presented in a somewhat monolithic manner:
parameterized sets of agents (residents and hospitals) are listed
along with the associated preference lists~\cite[Figure 1,
  erratum]{McDermid2010}. At the same time, like many transformations
from 3-SAT variants, the transformation admits a gadget-based
interpretation: for each variable (resp., clause) in the given $\SAT$
instance, the resulting $\HRSP$ instance includes a set of agents
inducing a copy of a variable (resp., clause) gadget. Moreover, the
variable and the clause gadgets are implicitly constructed from
smaller gadgets. The variable gadget is constructed from a smaller
gadget that admits exactly two stable assignments; in the present
paper, we refer to this smaller gadget as a bistable gadget. The
clause gadget is constructed using three copies of a smaller gadget,
one for each literal in the clause; in the present paper, we refer to
this smaller gadget as a literal gadget. The literal gadget is in turn
constructed from a smaller gadget that does not admit a stable
assignment; in the present paper, we refer to this smaller gadget as
an unstable gadget.

We find it useful to present our transformation in a modular and
explicitly gadget-based style. To aid the reader's intuition, we
present a diagram of each gadget using the digraph representation
proposed by Balinski and Ratier~\cite{Balinski1997} and subsequently
employed in various works (see,
e.g.,~\cite{Baiou2000,Baiou2002,Baiou2007}); in the present paper, we
refer to such a diagram as a preference digraph.
\iffullversion
To better illustrate the relationship between our gadgets and those
of~\cite{McDermid2010}, in Appendix~\ref{sec:hrs-gadgets}, we present
preference digraphs for the gadgets in~\cite{McDermid2010}.
\fi

We now highlight the main new ideas underlying our gadget-based
transformation from $\SAT$ to $\HRSC$. A key difference between our
setting and that of~\cite{McDermid2010} is that a blocking coalition
involving a hospital $h$ cannot propose to reduce the total size of
the set of residents assigned to $h$. This seems to necessitate a more
complicated construction than in~\cite{McDermid2010} for the unstable,
literal, and bistable gadgets.
\iffullversion
For the unstable (resp., literal,
bistable) gadget, one may compare the preference digraphs of
Figures~\ref{fig:unstable-gadget} and~\ref{fig:hrs-unstable-gadget}
(resp., Figures~\ref{fig:literal-gadget}
and~\ref{fig:hrs-literal-gadget}, Figures~\ref{fig:bistable-gadget}
and~\ref{fig:hrs-bistable-gadget}).

\fi
In developing our variable gadget from our bistable gadget as in the
preference digraph of Figure~\ref{fig:variable-gadget}, we noticed
that the corresponding construction in~\cite{McDermid2010} can be
streamlined via the same approach.
\iffullversion
The preference digraph of Figure~\ref{fig:hrs-variable-gadget} in
Appendix~\ref{sec:hrs-gadgets} shows the non-streamlined variable
gadget implicit in~\cite{McDermid2010}, which has two more residents
and two more hospitals than the streamlined version.
\else
The streamlined construction reduces the number of residents from 10
to 8 and the number of hospitals from 6 to 4.
\fi

While our literal gadget differs from that of~\cite{McDermid2010}, it
satisfies similar abstract properties such that the construction of
our clause gadget from three copies of the literal gadget can follow
precisely the same pattern as is implicit
in~\cite{McDermid2010}. Likewise, our variable and clause gadgets
satisfy similar abstract properties to those satisfied by the
corresponding gadgets of~\cite{McDermid2010}, such that our overall
construction of the $\HRSC$ instance corresponding to a given $\SAT$
instance follows the pattern implicit in~\cite{McDermid2010}. An
advantage of our modular gadget-based presentation is that individual
modules may be reusable in proofs of other similar results. In
Section~\ref{sec:unitwise}, we discuss a natural coalitional notion of
stability that is different from the two notions (first-coalition
stability and coalition stability) defined earlier, and we explain why
our NP-hardness result related to coalition stability also applies to
this third notion, which we call unitwise-coalition stability.

Like the NP-hardness proof of McDermid and Manlove for $\HRSP$, our
NP-hardness proof for $\HRSC$ involves $\HRS$ instances where the
resident sizes and hospital capacities belong to $\{1,2\}$. Because
the resident sizes belong to $\{1,2\}$, these NP-hardness results
apply to the important special case of $\HRS$ called
\emph{Hospitals/Residents with Inseparable Couples} ($\HRIC$).

As in McDermid and Manlove~\cite{McDermid2010}, our NP-hardness proof
involves $\HRSC$ instances with incomplete preferences, that is, where
only a subset of the residents (resp., hospitals) are acceptable to a
given hospital (resp., resident). Indeed, our construction only
requires preference lists of length 3 (resp., 4) for the residents
(resp., hospitals). The reader may wonder whether our NP-hardness
result continues to hold for $\HRSC$ instances with complete
preferences. In Section~\ref{sec:complete}, we give a simple reduction
showing that $\HRSC$ remains NP-hard in the special case of complete
preferences, even when the resident sizes and hospital capacities
continue to belong to $\{1,2\}$.

In the foregoing, we have focused our discussion of the prior
literature on results that are directly relevant to our NP-hardness
proof. For a broader overview of two-sided assignment problems with
preferences, we refer the reader to Rodr\'{i}guez and
Manlove~\cite{Rodriguez2025}.

\iffullversion
\else
Due to space limitations, some lemmas (Lemmas~\ref{lem:outward}
and~\ref{lem:inward}, and all lemmas in Section~\ref{sec:nph}) are
stated without proof; proofs may be found in the full
version~\cite{Hu2026}. The full version also includes two appendices:
Appendix~A provides a tabular representation of the construction of
Section~\ref{sec:transform}; Appendix~B presents preference digraphs
for the gadgets implicit in the construction
of~\cite[erratum]{McDermid2010}.
\fi

\section{Preliminaries}
\label{sec:prelim}

We now give a formal definition of the $\HRSC$ problem discussed in
Section~\ref{sec:intro}. The input is an $\HRS$ instance, which we now
describe. An $\HRS$ instance $I$ consists of a set of residents, where
each resident $r$ has a positive integer size, denoted $\Size{r}$, and
a set of hospitals, denoted $\Hosps{I}$, where each hospital $h$ has a
positive integer capacity, denoted $\Cap{h}$. For any set of residents
$R$, we define $\Size{R}$ as $\sum_{r \in R} \Size{r}$. Each resident
(resp., hospital) has a strict preference list over the set of
hospitals (resp., residents) that it considers to be acceptable. An
assignment $\Match$ for such an instance $I$ is a set of
resident-hospital pairs subject to certain constraints to be defined
below. If the resident-hospital pair $(r,h)$ belongs to an assignment
$\Match$, we say that $r$ is assigned to $h$ under $\Match$. As in
Section~\ref{sec:intro}, we define $\Res{\Match}{h}$ as $\{r\mid (r,h)
\in \Match\}$. An assignment $\Match$ is required to satisfy the
following constraints.
\begin{enumerate}
  \item Acceptability constraints. If $(r,h)$ belongs to $\Match$,
    then $h$ is acceptable to $r$ and $r$ is acceptable to $h$.
  \item Resident degree constraints. Any resident appears in at most
    one pair in $\Match$.
  \item Hospital capacity constraints. For any hospital $h$, we have
    $\Size{\Res{\Match}{h}} \le \Cap{h}$.
\end{enumerate}
We assume without loss of generality that for any resident $r$ and
hospital $h$, $r$ considers $h$ to be acceptable if and only if $h$
considers $r$ to be acceptable. 

Given such an $\HRS$ instance $I$ and an assignment $\Match$ for $I$,
we say that a set of residents $R$ and a hospital $h$ form a blocking
coalition $(R,h)$ for $\Match$ in $I$ if the following conditions
hold: (1) for all $r$ in $R$, $r$ is unassigned under $\Match$ or $r$
prefers $h$ to its assigned hospital under $\Match$, and (2) there
exists a subset $R'$ of $\{r' \in \Res{\Match}{h}\mid \mbox{$h$
  prefers $r$ to $r'$ for all $r$ in $R$}\}$ such that
\[
\Size{R} - \Cap{h} + \Size{\Res{\Match}{h}} \le \Size{R'} \le \Size{R}.
\]
For any $\HRS$ instance $I$ and any assignment $\Match$ for $I$, we
define $\Block{I}{\Match}$ as the set of all blocking coalitions for
$\Match$ in $I$, and we say that $\Match$ is coalition-stable for $I$
if $\Block{I}{\Match} = \emptyset$. Given an $\HRS$ instance $I$ as
input, the $\HRSC$ problem asks us to either find a coalition-stable
assignment for $I$, or to report correctly that no such assignment
exists. The decision version of $\HRSC$ asks whether a given $\HRS$
instance admits a coalition-stable assignment.

Our main result is to show that $\HRSC$ is NP-hard. We obtain this
result by presenting a polynomial-time transformation from the known
NP-complete problem $\SAT$ to the decision version of $\HRSC$. The
problem $\SAT$ is a special case of the well-known 3-SAT problem that
was shown to be NP-complete by Berman et al.~\cite[Theorem~1]{Berman}
(where $\SAT$ is referred to as (3,B2)-SAT). The input is a 3-CNF
formula, that is, a propositional formula that is the conjunction of a
number of clauses, where each clause is a disjunction of three
distinct literals, and each literal is a variable or the negation of a
variable. In addition, we require that each variable occurs exactly
four times, twice in negated form and twice in positive (i.e., not
negated) form. Given this requirement, any instance of $\SAT$ involves
$3n$ variables and $4n$ clauses for some positive integer $n$.

Let $\Prop$ be an instance of $\SAT$ and let $n$ be a positive integer
such that $\Prop$ has $3n$ variables and $4n$ clauses. We index the
variables from 1 to $3n$ and the clauses from $1$ to $4n$. For any $j$
in $\{1,\ldots,3n\}$ and any $k$ in $\{1,2,3,4\}$, we define literal
$(j,k)$ as follows: if $k = 1$ (resp., $k=2$) then it is the first
(resp., second) positive instance of variable $j$ in $\Prop$; if $k =
3$ (resp., $k=4$) then it is the first (resp., second) negative
instance of variable $j$ in $\Prop$. 

Our NP-hardness proof involves the construction of various $\HRSC$
``gadgets'', some of which are defined in terms of smaller
``sub-gadgets''. The following additional definitions, along with
Lemmas~\ref{lem:outward} and~\ref{lem:inward} below, are useful for
reasoning about the interface between a gadget and one of its
sub-gadgets. For any $\HRSC$ instance $I$ and any subset $H$ of the
hospitals in $I$, we define $\Proj{I}{H}$ as the $\HRSC$ instance
obtained from $I$ by removing all residents who do not have at least
one acceptable hospital in $H$, removing all hospitals not in $H$, and
pruning the removed hospitals from the preference lists of the
non-removed residents. For any $\HRSC$ instance $I$, any assignment
$\Match$ of $I$, and any subset $H$ of the hospitals in $I$, we define
$\Proj{\Match}{H}$ as the assignment $\{(r,h) \in \Match\mid h \in
H\}$ for $\Proj{I}{H}$, and we define $\Proj{B}{H}$ as $\{(R,h) \in
B\mid h \in H\}$ where $B$ denotes $\Block{I}{\Match}$.

\begin{lemma}
  \label{lem:outward}
  Let $I$ be an $\HRSC$ instance, let $\Match$ be an assignment of
  $I$, let $B$ denote $\Block{I}{\Match}$, let $H$ be a subset of
  $\Hosps{I}$, let $I'$ denote $\Proj{I}{H}$, let $B'$ denote
  $\Proj{B}{H}$, let $\Match'$ denote $\Proj{\Match}{H}$, and let
  $B''$ denote $\Block{I'}{\Match'}$. Then $B' \subseteq B''$.
\end{lemma}
\begin{deferproof}{Lemma~\ref{lem:outward}}
  Let $(R,h)$ be a blocking coalition in $B'$. Hence $h$ is a hospital
  in $H$ and $R$ is a set of residents in $I'$. Below we establish the
  claim of the lemma by proving that $(R,h)$ belongs to $B''$.

  We begin by arguing that condition~(1) in the definition of a
  blocking coalition in $B''$ holds for $(R,h)$. Let $r$ belong to
  $R$. If $r$ is unassigned in $\Match$ or $(r,h')$ belongs to
  $\Match$ where $h'$ belongs to $\Hosps{I} \setminus H$, then $r$ is
  unassigned in $\Match'$ and hence $r$ prefers $h$ to its assignment
  in $\Match'$, as required by condition~(1). If $(r,h')$ belongs to
  $\Match$ where $h'$ belongs to $H$, then $(r,h')$ belongs to
  $\Match'$, and since $(R,h)$ belongs to $B'$, we conclude that $r$
  prefers $h$ to $h'$, as required by condition~(1).
  
  It remains to argue that condition~(2) in the definition of a
  blocking coalition in $B''$ holds for $(R,h)$. This follows because
  the set of acceptable residents of $h$ is the same in $I$ and $I'$,
  and $\Res{\Match}{h} = \Res{\Match'}{h}$.
\end{deferproof}

For any $\HRSC$ instance $I$, any subset $H$ of $\Hosps{I}$, and any
assignment $\Match$ of $I$, we define $\Bdry{I}{H}{\Match}$ as the set
of all blocking coalitions $(R,h)$ in
$\Block{\Proj{I}{H}}{\Proj{\Match}{H}}$ for which there is a resident
$r$ in $R$ and a hospital $h'$ in $\Hosps{I} \setminus H$ such that
$r$ prefers $h'$ to $h$.

\begin{lemma}
  \label{lem:inward}
  Let $I$ be an $\HRSC$ instance, let $\Match$ be an assignment of
  $I$, let $B$ denote $\Block{I}{\Match}$, let $H$ be a subset of
  $\Hosps{I}$, let $I'$ denote $\Proj{I}{H}$, let $B'$ denote
  $\Proj{B}{H}$, let $\Match'$ denote $\Proj{\Match}{H}$, and let
  $B''$ denote $\Block{I'}{\Match'}$. Then $B''\subseteq B' \cup
  \Bdry{I}{H}{\Match}$.
\end{lemma}
\begin{deferproof}{Lemma~\ref{lem:inward}}
  Let $(R,h)$ be a blocking coalition in $B'' \setminus
  \Bdry{I}{H}{\Match}$ where $h$ is a hospital in $H$ and $R$ is a set
  of residents in $I'$. Below we establish the claim of the lemma by
  proving that $(R,h)$ belongs to $B'$. Since $h$ belongs to $H$, this
  is equivalent to proving that $(R,h)$ belongs to $B$.

  We begin by arguing that condition~(1) in the definition of a
  blocking coalition in $B$ holds for $(R,h)$. Let $r$ belong to
  $R$. We consider three cases.

  Case~1: $r$ is unassigned in $\Match$. Hence $r$ prefers $h$ to its
  assignment in $\Match$, as required by condition~(1).

  Case~2: $(r,h')$ belongs to $\Match$ where $h'$ belongs to
  $\Hosps{I} \setminus H$. Since $(R,h)$ does not belong to
  $\Bdry{I}{H}{\Match}$, we deduce that $r$ prefers $h$ to $h'$, as
  required by condition~(1).

  Case~3: $(r,h')$ belongs to $\Match$ where $h'$ belongs to
  $H$. Hence $(r,h')$ belongs to $\Match'$. Since $(R,h)$ belongs to
  $B''$, we conclude that $r$ prefers $h$ to $h'$, as required by
  condition~(1).

  It remains to argue that condition~(2) in the definition of a
  blocking coalition in $B$ holds for $(R,h)$. This follows because
  the set of acceptable residents of $h$ is the same in $I$ and $I'$,
  and $\Res{\Match}{h} = \Res{\Match'}{h}$.
\end{deferproof}

\section{Building Blocks}
\label{sec:blocks}

In this section, we present three gadgets. The first of these, which
we call the unstable gadget, is an $\HRSC$ instance with four
residents and three hospitals that does not admit a coalition-stable
assignment, thereby providing a negative answer to the question of
Rodr\'{i}guez and Manlove~\cite{Rodriguez2025} regarding the existence
of coalition-stable assignments. The second gadget that we present,
which we call the literal gadget, contains a copy of the unstable
gadget along with three additional residents and two additional
hospitals. (One resident in the unstable gadget has its preference
list extended to include one of the two additional hospitals.) In
Section~\ref{sec:clause}, we define a clause gadget that includes
three copies of the literal gadget, one for each of the associated
literals. The third gadget that we present, which we call the bistable
gadget, has five residents and three hospitals and admits exactly two
coalition-stable assignments. In Section~\ref{sec:variable}, we define
a variable gadget that includes a copy of the bistable gadget.

Our unstable, literal, and bistable gadgets are analogous to fragments
of the monolithic construction presented by McDermid and
Manlove~\cite{McDermid2010} in the context of pair stability. Each of
these gadgets differs from the corresponding fragment in its detailed
design.\footnote{Using the notation of~\cite[Figure~1,
  erratum]{McDermid2010}, the correspondence is as follows: the
subinstance induced by the set of residents $\{q_{j,1}^k, q_{j,2}^k,
q_{j,3}^k\}$ and the set of hospitals $\{p_{j,1}^k, p_{j,2}^k\}$ is
analogous to our unstable gadget; the subinstance induced by the set
of residents $\{q_{j,1}^k, q_{j,2}^k, q_{j,3}^k, v(p_{j,3}^k)\}$ and
the set of hospitals $\{p_{j,1}^k, p_{j,2}^k, p_{j,3}^k\}$ is
analogous to our literal gadget; the subinstance induced by the set of
residents $\{r_j^1, r_j^2, r_j^3, r_j^4\}$ and the set of hospitals
$\{h_j^1, h_j^2\}$ is analogous to our bistable gadget.} Throughout
the remainder of this section, the only notion of stability we are
concerned with is coalition stability. Accordingly, we find it
convenient to use the term ``stability'' (resp., ``stable'') to mean
``coalition stability'' (resp., ``coalition-stable''). For any $\HRSC$
instance $I$, any assignment $\Match$ of $I$, and any subset $H$ of
$\Hosps{I}$, we say that $\Match$ is $H$-stable if no blocking
coalition in $\Block{I}{\Match}$ involves a hospital in $H$.

\subsection{Unstable Gadget}
\label{sec:unstable}

In this section, we present an $\HRSC$ instance that does not admit a
stable assignment, thereby settling (in the negative) the open
question of Rodr\'{i}guez and Manlove~\cite{Rodriguez2025} concerning
whether a stable assignment always exists.

\begin{figure}
    \centering
    \begin{minipage}[b]{0.65\textwidth}
        \centering
        \begin{picture}(88.00,123.20)
\put(24.00,123.20){\makebox(0,0){$\sigma_1$}}
\put(56.00,123.20){\makebox(0,0){\fbox{$\sigma_2$}}}
\put(88.00,123.20){\makebox(0,0){$\sigma_3$}}
\put(0.00,3.20){\makebox(0,0){$s_4$}}
\put(0.00,35.20){\makebox(0,0){\fbox{$s_3$}}}
\put(0.00,67.20){\makebox(0,0){$s_2$}}
\put(0.00,99.20){\makebox(0,0){$s_1$}}
\put(56.00,35.20){\circle*{6.40}}
\put(56.00,8.00){\vector(0,1){22.40}}
\put(56.00,67.20){\circle*{6.40}}
\put(56.00,40.00){\vector(0,1){22.40}}
\put(56.00,99.20){\circle*{6.40}}
\put(56.00,72.00){\vector(0,1){22.40}}
\put(24.00,99.20){\circle*{6.40}}
\put(51.20,99.20){\vector(-1,0){22.40}}
\put(24.00,67.20){\circle*{6.40}}
\put(24.00,94.40){\vector(0,-1){22.40}}
\put(56.00,67.20){\circle*{6.40}}
\put(28.80,67.20){\vector(1,0){22.40}}
\put(88.00,67.20){\circle*{6.40}}
\put(60.80,67.20){\vector(1,0){22.40}}
\put(88.00,3.20){\circle*{6.40}}
\put(88.00,62.40){\vector(0,-1){54.40}}
\put(56.00,3.20){\circle*{6.40}}
\put(83.20,3.20){\vector(-1,0){22.40}}
\end{picture}
        \subcaption{The unstable gadget.}
        \label{fig:unstable-gadget}
    \end{minipage}%
    \begin{minipage}[b]{0.3\textwidth}
        \centering
        \begin{picture}(36.80,40.00)
\put(24.00,32.00){\makebox(0,0){\fbox{$\sigma_2$}}}
\put(0.00,8.00){\makebox(0,0){\fbox{$s_3$}}}
\put(24.00,8.00){\circle*{6.40}}
\put(24.00,8.00){\line(2,1){16.00}}
\put(24.00,8.00){\line(2,-1){16.00}}
\put(40.00,0.00){\line(0,1){16.00}}
\put(28.00,8.00){\line(2,1){10.56}}
\put(28.00,8.00){\line(2,-1){10.56}}
\put(38.56,2.72){\line(0,1){10.56}}
\end{picture}
        \subcaption{The associated icon.}
    \end{minipage}
    \caption{The preference digraph for the unstable gadget has four
      residents $s_1, \ldots, s_4$ and three hospitals $\sigma_1,
      \sigma_2, \sigma_3$ with sizes and capacities as indicated. The
      unstable gadget icon is used in the preference digraph of
      Figure~\ref{fig:literal-gadget}; the dot in this icon
      corresponds to the resident-hospital pair $(s_3, \sigma_2)$.}
\end{figure}
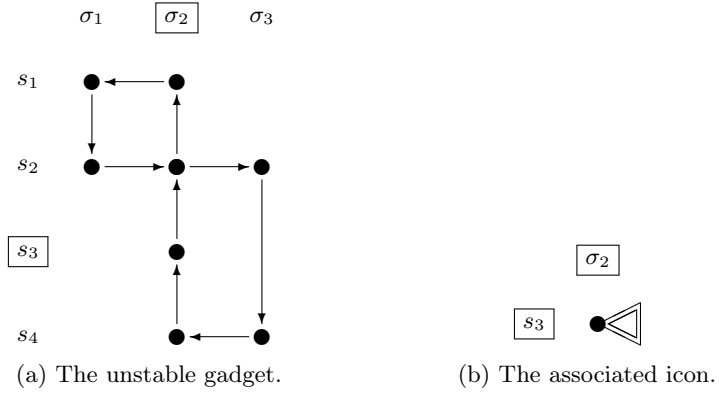

As discussed in Section~\ref{sec:intro}, we find it useful to present
all of our gadgets using a preference
digraph. Figure~\ref{fig:unstable-gadget} presents the preference
digraph for our unstable gadget: each dot represents a mutually
acceptable resident-hospital pair; dots in the same row (resp.,
column) have the same associated resident (resp., hospital); the
symbol associated with each row (resp., column) represents the
associated resident (resp., hospital); if the symbol is enclosed in a
box then the size (resp., capacity) of the resident (resp., hospital)
is 2, otherwise it is 1; the directed edges within a row (resp.,
column) are used to represent the preference list of the associated
resident (resp., hospital). To understand how the directed edges
determine preference order, it is sufficient to consider the example
of hospital $\sigma_2$ in Figure~\ref{fig:unstable-gadget}. Hospital
$\sigma_2$ has four acceptable residents $s_1$, $s_2$, $s_3$, and
$s_4$. These four residents are linked together (within column
$\sigma_2$) by three directed edges forming a simple path from $s_4$
to $s_3$ to $s_2$ to $s_1$. This path signifies that $s_1$ (resp.,
$s_2$, $s_3$, $s_4$) is the first (resp., second, third, fourth)
preference of $\sigma_2$. Similarly, resident $s_2$ has three
acceptable hospitals $\sigma_1$, $\sigma_2$, and $\sigma_3$, where
$\sigma_3$ (resp., $\sigma_2$, $\sigma_1$) is the first (resp.,
second, third) preference of $s_2$.

Throughout the remainder of Section~\ref{sec:unstable}, let $I$ denote
the unstable gadget of Figure~\ref{fig:unstable-gadget}. We define the
\emph{critical assignment of $I$} as $\{(s_1, \sigma_1), (s_2,
\sigma_3), (s_4, \sigma_2) \}$.

\begin{lemma}
\label{lem:unstable}
  Instance $I$ does not admit a stable assignment, and the critical
  assignment of $I$ is the unique assignment $\Match$ of $I$ such
  that $\Block{I}{\Match} = \{(s_3, \sigma_2)\}$.
\end{lemma}
\begin{proof}
  Let $\Match^*$ denote the critical assignment of $I$ and let $B^*$
  denote $\Block{I}{\Match^*}$. We first argue that $B^* = \{(s_3,
  \sigma_2)\}$. It is easy to see that $(s_3, \sigma_2)$ belongs to
  $B^*$ and is the unique blocking coalition in $B^*$ that involves
  resident $s_3$. Each resident in the set $\{s_1, s_2, s_4\}$ gets
  its most preferred hospital and so cannot participate in any
  blocking coalition. Thus $B^* = \{(s_3, \sigma_2)\}$, as required.

  Let $\Match$ be an assignment for $I$ such that $\Block{I}{\Match}
  \subseteq \{(s_3, \sigma_2)\}$. Below we complete the proof by
  proving that $\Match = \Match^*$. Let $B$ denote
  $\Block{I}{\Match}$.

  We begin by arguing that $(s_1, \sigma_2)$ does not belong to
  $\Match$. Assume for the sake of contradiction that $(s_1,
  \sigma_2)$ belongs to $\Match$. Hence $(s_3, \sigma_2)$ does not
  belong to $\Match$. In addition, since $(s_1, \sigma_1)$ does not
  belong to $B$, we deduce that $(s_2, \sigma_1)$ belongs to
  $\Match$. Hence $(s_2, \sigma_2)$ does not belong to $\Match$
  and we conclude that $\Res{\Match}{\sigma_2} \subseteq \{s_1,
  s_4\}$. It follows that $(s_2, \sigma_2)$ belongs to $B$, a
  contradiction.

  We now argue that $(s_3, \sigma_2)$ does not belong to
  $\Match$. Assume for the sake of contradiction that $(s_3,
  \sigma_2)$ belongs to $\Match$. Hence $\Res{\Match}{\sigma_2} =
  \{s_3\}$. Since $(s_4, \sigma_3)$ does not belong to $B$, we deduce
  that $(s_4, \sigma_3)$ belongs to $\Match$. Since $(s_2, \sigma_1)$
  does not belong to $B$, we deduce that $(s_2, \sigma_1)$ belongs to
  $\Match$. It follows that $s_1$ is unassigned in $\Match$ and hence
  $(\{s_1, s_2\}, \sigma_2)$ belongs to $B$, a contradiction.

  Since $(s_1, \sigma_2)$ and $(s_3, \sigma_2)$ do not belong to
  $\Match$, we conclude that $\Res{\Match}{\sigma_2} \subseteq \{s_2,
  s_4\}$. Since $(s_4, \sigma_2)$ does not belong to $B$, we deduce
  that $(s_4, \sigma_2)$ belongs to $\Match$. Hence
  $\Res{\Match}{\sigma_3} \subseteq \{s_2\}$. Since $(s_2, \sigma_3)$
  does not belong to $B$, we deduce that $(s_2, \sigma_3)$ belongs to
  $\Match$. Hence $\Res{\Match}{\sigma_1} \subseteq \{s_1\}$. Since
  $(s_1, \sigma_1)$ does not belong to $B$, we deduce that $(s_1,
  \sigma_1)$ belongs to $\Match$. In summary, $\Match = \Match^*$, as
  required.
\end{proof}

\subsection{Literal Gadget}
\label{sec:literal}

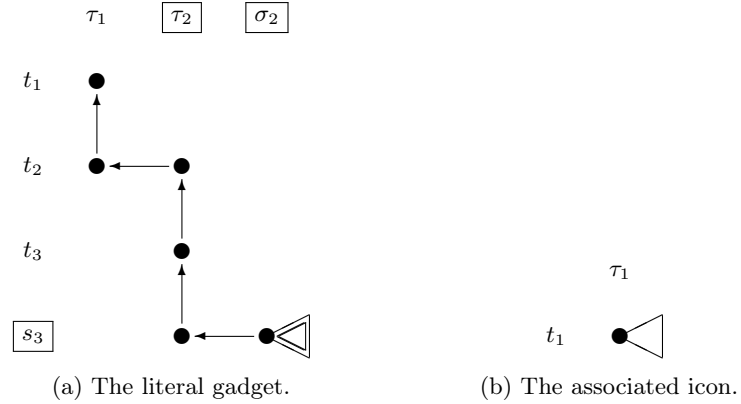
\begin{figure}
    \centering
    \begin{minipage}[b]{0.65\textwidth}
        \centering
        \begin{picture}(100.80,128.00)
\put(24.00,128.00){\makebox(0,0){$\tau_1$}}
\put(56.00,128.00){\makebox(0,0){\fbox{$\tau_2$}}}
\put(88.00,128.00){\makebox(0,0){\fbox{$\sigma_2$}}}
\put(0.00,8.00){\makebox(0,0){\fbox{$s_3$}}}
\put(0.00,40.00){\makebox(0,0){$t_3$}}
\put(0.00,72.00){\makebox(0,0){$t_2$}}
\put(0.00,104.00){\makebox(0,0){$t_1$}}
\put(88.00,8.00){\circle*{6.40}}
\put(88.00,8.00){\line(2,1){16.00}}
\put(88.00,8.00){\line(2,-1){16.00}}
\put(104.00,0.00){\line(0,1){16.00}}
\put(92.00,8.00){\line(2,1){10.56}}
\put(92.00,8.00){\line(2,-1){10.56}}
\put(102.56,2.72){\line(0,1){10.56}}
\put(83.20,8.00){\vector(-1,0){22.40}}
\put(56.00,8.00){\circle*{6.40}}
\put(56.00,12.80){\vector(0,1){22.40}}
\put(56.00,40.00){\circle*{6.40}}
\put(56.00,44.80){\vector(0,1){22.40}}
\put(56.00,72.00){\circle*{6.40}}
\put(51.20,72.00){\vector(-1,0){22.40}}
\put(24.00,72.00){\circle*{6.40}}
\put(24.00,76.80){\vector(0,1){22.40}}
\put(24.00,104.00){\circle*{6.40}}
\end{picture}
        \subcaption{The literal gadget.}
        \label{fig:literal-gadget}
    \end{minipage}%
    \begin{minipage}[b]{0.3\textwidth}
        \centering
        \begin{picture}(36.80,32.00)
\put(24.00,32.00){\makebox(0,0){$\tau_1$}}
\put(0.00,8.00){\makebox(0,0){$t_1$}}
\put(24.00,8.00){\circle*{6.40}}
\put(24.00,8.00){\line(2,1){16.00}}
\put(24.00,8.00){\line(2,-1){16.00}}
\put(40.00,0.00){\line(0,1){16.00}}
\end{picture}
        \subcaption{The associated icon.}
    \end{minipage}
    \caption{The preference digraph for the literal gadget includes a
      copy of the unstable gadget and has a total of $3 + 4 = 7$
      residents and $2 + 3 = 5$ hospitals. Within the literal gadget,
      resident $s_3$ has two acceptable hospitals $\tau_2$ and
      $\sigma_2$, with $\tau_2$ being the first preference of $s_3$ as
      indicated by the directed edge in row $s_3$.
      \iffullversion
      The literal gadget icon is used in the preference digraphs of
      Figures~\ref{fig:clause-gadget}
      and~\ref{fig:global-construction}; the dot in this icon refers
      to the resident-hospital pair $(t_1, \tau_1)$.
      \else
      The literal gadget icon is used in the preference digraph of
      Figure~\ref{fig:clause-gadget}; the dot in this icon refers
      to the resident-hospital pair $(t_1, \tau_1)$.
      \fi}
\end{figure}

Throughout the remainder of Section~\ref{sec:literal}, let $I$ denote
the literal gadget of Figure~\ref{fig:literal-gadget} and let $I'$
denote the unstable gadget of $I$. We define the \emph{critical
assignment of $I$} as the union of the critical assignment of $I'$ and
$\{(s_3, \tau_2), (t_2, \tau_1)\}$.

\begin{lemma}
\label{lem:literal}
  Instance $I$ does not admit a stable assignment, and the critical
  assignment of $I$ is the unique assignment $\Match$ of $I$ such that
  $\Block{I}{\Match} = \{(t_1, \tau_1)\}$.
\end{lemma}
\begin{proof}
  Let $\Match^*$ denote the critical assignment of $I$, let $B^*$
  denote $\Block{I}{\Match^*}$, let $H$ denote the set of hospitals in
  $I'$, let $\Match'$ denote $\Proj{\Match^*}{H}$, and let $B'$ denote
  $\Block{I'}{\Match'}$. Thus $\Match'$ is the critical assignment of
  $I'$, and hence Lemma~\ref{lem:unstable} implies $B' = \{(s_3,
  \sigma_2)\}$.

  We claim that $B^* = \{(t_1, \tau_1)\}$. It is easy to see that
  $(t_1, \tau_1)$ belongs to $B^*$. Since resident $t_1$ has only one
  acceptable hospital $\tau_1$, and since $t_1$ has size 1 and
  $\tau_1$ has capacity 1, resident $t_1$ cannot belong to any other
  blocking coalition in $B^*$ apart from $(t_1, \tau_1)$. Since
  residents $s_3$ and $t_2$ are each assigned to their most preferred
  hospital under $\Match^*$, neither $s_3$ nor $t_2$ participates in
  any blocking coalition in $B^*$. Since resident $t_3$ has size 1,
  resident $s_3$ has size 2, and $(s_3, \tau_2)$ belongs to
  $\Match^*$, we deduce that resident $t_3$ does not participate in
  any blocking coalition in $B^*$. Given the foregoing observations,
  it is sufficient to argue that $\Match^*$ is
  $H$-stable. Lemma~\ref{lem:outward} implies
  $\Proj{B^*}{H} \subseteq B'$. Since $B' = \{(s_3, \sigma_2)\}$ and
  since resident $s_3$ does not participate in any blocking coalition
  in $B^*$, we deduce that $\Match^*$ is $H$-stable.

  Let $\Match_1$ be an assignment of $I$ such that
  $\Block{I}{\Match_1} \subseteq \{(t_1, \tau_1)\}$ and let $B$ denote
  $\Block{I}{\Match_1}$. Below we complete the proof by establishing
  that $\Match_1 = \Match^*$.

  \iffullversion
      We claim that $\Bdry{I}{H}{\Match_1} \subseteq \{(s_3,
      \sigma_2)\}$. Let $(R,h)$ belong to $\Bdry{I}{H}{\Match_1}$. Since
      $s_3$ is the only resident in $I'$ with an acceptable hospital in
      $\Hosps{I} \setminus H$, we deduce that $s_3$ belongs to $R$. Since
      $\sigma_2$ is the only acceptable hospital of $s_3$ in $I'$, we
      deduce that $h = \sigma_2$. Since $\Size{s_3} = 2 = \Cap{\sigma_2}$,
      the claim follows.
  \fi

  Let $B''$ denote $\Block{I'}{\Proj{\Match_1}{H}}$. Since
  $\Bdry{I}{H}{\Match_1} \subseteq \{(s_3, \sigma_2)\}$,
  Lemma~\ref{lem:inward} implies
  $B'' \setminus \{(s_3, \sigma_2)\} \subseteq \Proj{B}{H} \subseteq
  B$. Since $B \subseteq \{(t_1, \tau_1)\}$, we conclude from
  Lemma~\ref{lem:unstable} that $\Proj{\Match_1}{H} = \Match'$. Hence
  $B'' = \{(s_3, \sigma_2)\}$. Since $(s_3, \sigma_2)$ does not belong
  to $B$, we deduce that $(s_3, \tau_2)$ belongs to $\Match_1$. It
  follows that $\Res{\Match_1}{\tau_2} = \{s_3\}$. Since
  $(\{t_2, t_3\}, \tau_2)$ does not belong to $B$, we deduce that
  $(t_2, \tau_1)$ belongs to $\Match_1$. Thus
  $\Match_1 = \Match' \cup \{(s_3, \tau_2), (t_2, \tau_1)\} =
  \Match^*$.
\end{proof}

\subsection{Bistable Gadget}
\label{sec:bistable}

\begin{figure}
    \centering
    \begin{minipage}[b]{0.65\textwidth}
        \centering
        \begin{picture}(88.00,155.20)
\put(24.00,155.20){\makebox(0,0){\fbox{$\alpha_1$}}}
\put(56.00,155.20){\makebox(0,0){$\alpha_2$}}
\put(88.00,155.20){\makebox(0,0){\fbox{$\alpha_3$}}}
\put(0.00,3.20){\makebox(0,0){$a_5$}}
\put(0.00,35.20){\makebox(0,0){\fbox{$a_4$}}}
\put(0.00,67.20){\makebox(0,0){$a_3$}}
\put(0.00,99.20){\makebox(0,0){\fbox{$a_2$}}}
\put(0.00,131.20){\makebox(0,0){$a_1$}}
\put(24.00,3.20){\circle*{6.40}}
\put(28.80,3.20){\vector(1,0){54.40}}
\put(88.00,3.20){\circle*{6.40}}
\put(88.00,8.00){\vector(0,1){22.40}}
\put(88.00,35.20){\circle*{6.40}}
\put(88.00,40.00){\vector(0,1){86.40}}
\put(88.00,131.20){\circle*{6.40}}
\put(83.20,131.20){\vector(-1,0){54.40}}
\put(24.00,131.20){\circle*{6.40}}
\put(24.00,126.40){\vector(0,-1){22.40}}
\put(24.00,99.20){\circle*{6.40}}
\put(24.00,94.40){\vector(0,-1){22.40}}
\put(24.00,67.20){\circle*{6.40}}
\put(51.20,67.20){\vector(-1,0){22.40}}
\put(56.00,67.20){\circle*{6.40}}
\put(24.00,62.40){\vector(0,-1){54.40}}
\end{picture}
        \subcaption{The bistable gadget.}
        \label{fig:bistable-gadget}
    \end{minipage}%
    \begin{minipage}[b]{0.3\textwidth}
        \centering
        \begin{picture}(56.00,68.00)
\put(24.00,72.00){\makebox(0,0){\fbox{$\alpha_1$}}}
\put(56.00,72.00){\makebox(0,0){\fbox{$\alpha_3$}}}
\put(0.00,8.00){\makebox(0,0){\fbox{$a_4$}}}
\put(0.00,40.00){\makebox(0,0){\fbox{$a_2$}}}
\put(24.00,0.00){\line(1,0){32.00}}
\put(24.00,48.00){\line(1,0){32.00}}
\put(24.00,0.00){\line(0,1){48.00}}
\put(56.00,0.00){\line(0,1){48.00}}
\put(24.00,40.00){\circle*{6.40}}
\put(56.00,8.00){\circle*{6.40}}
\end{picture}
        \subcaption{The associated icon.}
        \label{fig:bistable-icon}
    \end{minipage}
    \caption{The preference digraph for the bistable gadget has five
      residents $a_1, \ldots, a_5$ and three hospitals $\alpha_1,
      \alpha_2, \alpha_3$. The bistable gadget icon is used in the
      preference digraph of Figure~\ref{fig:variable-gadget}; the left
      (resp., right) dot in this icon corresponds to the
      resident-hospital pair $(a_2, \alpha_1)$ (resp., $(a_4,
      \alpha_3)$).}
\end{figure}
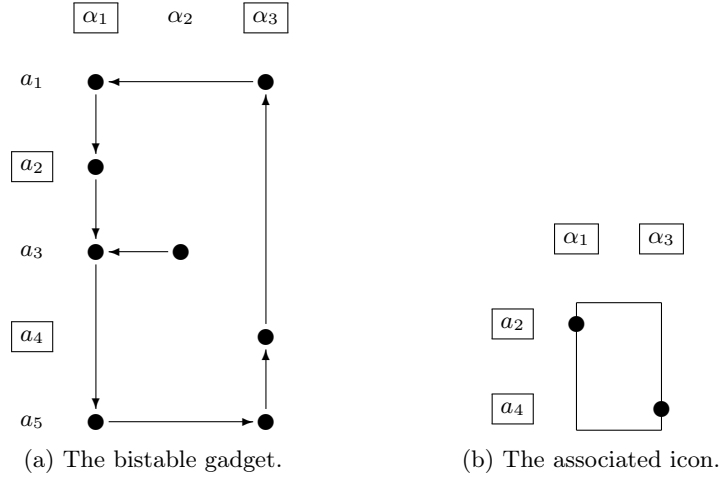

Throughout the remainder of Section~\ref{sec:bistable}, let $I$ denote
the bistable gadget depicted in Figure~\ref{fig:bistable-gadget}. We
define the \emph{true (resp., false) assignment of $I$} as $\{(a_1,
\alpha_3), (a_2, \alpha_1), (a_3, \alpha_2), (a_5, \alpha_3)\}$
(resp., $\{(a_3, \alpha_1), (a_4, \alpha_3), (a_5, \alpha_1)\}$).

\begin{lemma}
\label{lem:bistable}
    The true and false assignments of $I$ are stable, and $I$ does not
    admit any other stable assignments.
\end{lemma}
\begin{proof}
  Let $\Match_1$ (resp., $\Match_2$) denote the true (resp., false)
  assignment of $I$. We begin by arguing that $\Match_1$ is stable for
  $I$. Residents $a_2$ and $a_5$ are each assigned to their most
  preferred hospital, and hence cannot belong to a blocking coalition
  for $\Match_1$. Resident $a_1$ cannot belong to a blocking coalition
  for $\Match_1$ because hospital $\alpha_1$ prefers $a_2$ to
  $a_1$. Because $a_2$ has size 2, it is straightforward to argue that
  any blocking coalition involving $a_3$ and $\alpha_1$ has to be
  $(\{a_3, a_5\}, \alpha_1)$. Since resident $a_5$ cannot belong to a
  blocking coalition for $\Match_1$, we conclude that resident $a_3$
  cannot belong to a blocking coalition for $\Match_1$. Resident $a_4$
  cannot belong to a blocking coalition for $\Match_1$ because
  hospital $\alpha_3$ prefers $a_1$ to $a_4$.

  Now we argue that $\Match_2$ is stable for $I$. Hospital $\alpha_1$
  cannot belong to a blocking coalition for $\Match_2$ because
  $\Res{\Match_2}{\alpha_1} = \{a_3, a_5\}$ contains the two most
  preferred residents of $\alpha_1$. Hospital $\alpha_2$ cannot belong
  to a blocking coalition for $\Match_2$ because its lone acceptable
  resident $a_3$ is assigned to its most preferred hospital
  $\alpha_1$. Hospital $\alpha_3$ cannot belong to a blocking
  coalition for $\Match_2$ because resident $a_4$ has size two and the
  total size of all residents that $\alpha_3$ prefers to $a_4$ is only
  one.

  It remains to argue that $\Match_1$ and $\Match_2$ are the only
  stable assignments for $I$. Let $\Match$ be a stable assignment for
  $I$ and let $B$ denote $\Block{I}{\Match}$. Since $\Match$ is
  stable, $B$ is empty. We consider two cases.

  Case~1: $(a_3, \alpha_2)$ belongs to $\Match$. Since $(a_3,
  \alpha_1)$ does not belong to $B$, the Case~1 condition implies that
  $\Res{\Match}{\alpha_1}$ is either $\{a_2\}$ or $\{a_1,
  a_5\}$. Since $(a_3, \alpha_1)$ does not belong to $B$, the Case~1
  condition implies $(a_1, \alpha_1)$ does not belong to $\Match$ and
  hence $\Res{\Match}{\alpha_1} = \{a_2\}$. Since $(\{a_3, a_5\},
  \alpha_1)$ does not belong to $B$, we deduce that $(a_5, \alpha_3)$
  belongs to $\Match$. Since $(a_1, \alpha_3)$ does not belong to $B$,
  we deduce that $(a_1, \alpha_3)$ belongs to $\Match$. In summary,
  under the Case~1 condition, $\Match_1 \subseteq \Match$. Since
  hospitals $\alpha_1$, $\alpha_2$, and $\alpha_3$ are full under
  $\Match_1$, we conclude that $\Match = \Match_1$.

  Case~2: $(a_3, \alpha_2)$ does not belong to $\Match$. Since $(a_3,
  \alpha_2)$ does not belong to $B$, we deduce that $(a_3, \alpha_1)$
  belongs to $\Match$ and hence $(a_2, \alpha_1)$ does not belong to
  $\Match$. Thus $\Res{\Match}{\alpha_1} \subseteq \{a_1, a_3, a_5\}$.

  We first argue that $(a_5, \alpha_1)$ belongs to $\Match$. Assume
  for the sake of contradiction that $(a_5, \alpha_1)$ does not belong
  to $\Match$. Hence $\Res{\Match}{\alpha_1} \subseteq \{a_1,
  a_3\}$. Since $(a_1, \alpha_1)$ does not belong to $B$, we deduce
  that $(a_1, \alpha_1)$ belongs to $\Match$. Since $(a_4, \alpha_3)$
  does not belong to $B$, we deduce that $(a_4, \alpha_3)$ belongs to
  $\Match$. Thus $(a_5, \alpha_3)$ does not belong to
  $\Match$. Together with our assumption that $(a_5, \alpha_1)$ does
  not belong to $\Match$, this implies $a_5$ is unassigned in
  $\Match$. Since $\Res{\Match}{\alpha_1} \subseteq \{a_1, a_3\}$, we
  conclude that $(a_5, \alpha_1)$ belongs to $B$, a contradiction.

  Since $(a_3, \alpha_1)$ and $(a_5, \alpha_1)$ belong to $\Match$, we
  conclude that $\Res{\Match}{\alpha_1} = \{a_3, a_5\}$.

  We claim that $(a_1, \alpha_3)$ does not belong to $\Match$. Assume
  for the sake of contradiction that $(a_1, \alpha_3)$ belongs to
  $\Match$. Thus $(a_4, \alpha_3)$ does not belong to $\Match$. Hence
  $(a_5, \alpha_3)$ belongs to $B$, a contradiction.

  Since $(a_1, \alpha_3)$ does not belong to $\Match$ and $(a_4,
  \alpha_3)$ does not belong to $B$, we deduce that $(a_4, \alpha_3)$
  belongs to $\Match$. In summary, under the Case~2 condition,
  $\Match_2 \subseteq \Match$ and $\Res{\Match}{\alpha_2} =
  \emptyset$. Since hospitals $\alpha_1$ and $\alpha_3$ are full under
  $\Match_2$, we conclude that $\Match = \Match_2$.
\end{proof}

\section{NP-Hardness}
\label{sec:nph}

In the present section, we use a reduction from $\SAT$ to prove that
the decision version of $\HRSC$ is NP-hard. Since course allocation
generalizes $\HRS$, we conclude that the problem of finding a
coalition-stable assignment for a given instance of course allocation
is NP-hard, thereby resolving an open problem of Rodr\'{i}guez and
Manlove~\cite{Rodriguez2025}.

As discussed earlier, the literal gadget of Section~\ref{sec:literal}
satisfies similar abstract properties to those satisfied by the
corresponding fragment of the construction of McDermid and
Manlove~\cite{McDermid2010}. Consequently, we can use the approach
implicit in~\cite{McDermid2010} to build a clause gadget from three
copies of the literal gadget. Similarly, the bistable gadget can be
used to build a variable gadget using the approach implicit
in~\cite{McDermid2010}.\footnote{In studying the variable gadget
implicit in~\cite{McDermid2010}, which is induced by the residents in
the set $\{r_j^1, r_j^2, r_j^3, r_j^4, r_j^5, r_j^6, x_j^1, x_j^2,
y_j^1, y_j^2\}$ and the hospitals in the set $\{h_j^1, h_j^2, h_j^3,
h_j^4, h_j^T, h_j^F\}$, we observed that their construction can be
simplified. Specifically, the set of residents can be reduced to
$\{r_j^1, r_j^2, r_j^3, r_j^4, x_j^1, x_j^2, y_j^1, y_j^2\}$ and the
set of hospitals can be reduced to $\{h_j^1, h_j^2, h_j^T,
h_j^F\}$. Figure~\ref{fig:variable-gadget} shows how we construct our
variable gadget from our bistable gadget, and the same approach can be
used to construct a streamlined version of the variable gadget
implicit in~\cite{McDermid2010} from the bistable gadget implicit
in~\cite{McDermid2010}.} Finally, again following the approach
implicit in~\cite{McDermid2010}, we can build a suitable overall
$\HRSC$ instance for a given $\SAT$ formula $\Prop$ using a copy of
the variable gadget for each variable in $\Prop$ and a copy of the
clause gadget for each clause in $\Prop$.

For the sake of completeness, below we provide all of the details
concerning the construction of the variable gadget, the clause gadget,
and the overall $\HRSC$ instance corresponding to a given $\SAT$
formula $\Prop$. We emphasize that all of the constructions and
arguments in this section precisely mirror elements implicit in the
work of McDermid and Manlove~\cite{McDermid2010}.
As in Section~\ref{sec:blocks}, throughout this section we use the
term ``stability'' (resp., ``stable'') to mean ``coalition stability''
(resp., ``coalition-stable'').

\subsection{Variable Gadget}
\label{sec:variable}

\iffullversion
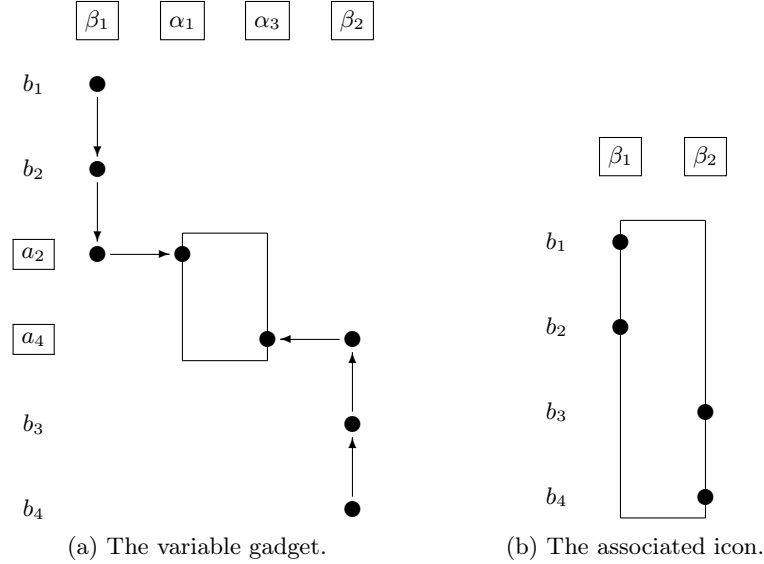
\begin{figure}
    \centering
    \begin{minipage}[b]{0.65\textwidth}
        \centering
        \begin{picture}(120.00,187.20)
\put(24.00,187.20){\makebox(0,0){\fbox{$\beta_1$}}}
\put(56.00,187.20){\makebox(0,0){\fbox{$\alpha_1\vphantom{\beta_1}$}}}
\put(88.00,187.20){\makebox(0,0){\fbox{$\alpha_3\vphantom{\beta_1}$}}}
\put(120.00,187.20){\makebox(0,0){\fbox{$\beta_2$}}}
\put(0.00,3.20){\makebox(0,0){$b_4$}}
\put(0.00,35.20){\makebox(0,0){$b_3$}}
\put(0.00,67.20){\makebox(0,0){\fbox{$a_4$}}}
\put(0.00,99.20){\makebox(0,0){\fbox{$a_2$}}}
\put(0.00,131.20){\makebox(0,0){$b_2$}}
\put(0.00,163.20){\makebox(0,0){$b_1$}}
\put(56.00,59.20){\line(1,0){32.00}}
\put(56.00,107.20){\line(1,0){32.00}}
\put(56.00,59.20){\line(0,1){48.00}}
\put(88.00,59.20){\line(0,1){48.00}}
\put(56.00,99.20){\circle*{6.40}}
\put(88.00,67.20){\circle*{6.40}}
\put(28.80,99.20){\vector(1,0){22.40}}
\put(24.00,99.20){\circle*{6.40}}
\put(24.00,126.40){\vector(0,-1){22.40}}
\put(24.00,131.20){\circle*{6.40}}
\put(24.00,158.40){\vector(0,-1){22.40}}
\put(24.00,163.20){\circle*{6.40}}
\put(115.20,67.20){\vector(-1,0){22.40}}
\put(120.00,67.20){\circle*{6.40}}
\put(120.00,40.00){\vector(0,1){22.40}}
\put(120.00,35.20){\circle*{6.40}}
\put(120.00,8.00){\vector(0,1){22.40}}
\put(120.00,3.20){\circle*{6.40}}
\end{picture}
        \subcaption{The variable gadget.}
        \label{fig:variable-gadget}
    \end{minipage}%
    \begin{minipage}[b]{0.3\textwidth}
        \centering
        \begin{picture}(56.00,132.00)
\put(24.00,136.00){\makebox(0,0){\fbox{$\beta_1$}}}
\put(56.00,136.00){\makebox(0,0){\fbox{$\beta_2$}}}
\put(0.00,8.00){\makebox(0,0){$b_4$}}
\put(0.00,40.00){\makebox(0,0){$b_3$}}
\put(0.00,72.00){\makebox(0,0){$b_2$}}
\put(0.00,104.00){\makebox(0,0){$b_1$}}
\put(24.00,0.00){\line(1,0){32.00}}
\put(24.00,112.00){\line(1,0){32.00}}
\put(24.00,0.00){\line(0,1){112.00}}
\put(56.00,0.00){\line(0,1){112.00}}
\put(24.00,104.00){\circle*{6.40}}
\put(24.00,72.00){\circle*{6.40}}
\put(56.00,40.00){\circle*{6.40}}
\put(56.00,8.00){\circle*{6.40}}
\end{picture}
        \subcaption{The associated icon.}
    \end{minipage}
    \caption{The preference digraph for the variable gadget includes a
      copy of the bistable gadget and has a total of $4 + 5 = 9$
      residents and $2 + 3 = 5$ hospitals. The variable gadget icon is
      used in the preference digraph of
      Figure~\ref{fig:global-construction}; the upper-left (resp.,
      lower-left, upper-right, lower-right) dot in this icon refers
      to the resident-hospital pair $(b_1, \beta_1)$ (resp., $(b_2,
      \beta_1)$, $(b_3, \beta_2)$, $(b_4, \beta_2)$).}
\end{figure}
\else
\fi

Throughout the remainder of Section~\ref{sec:variable}, let $I$
denote the variable gadget depicted in
Figure~\ref{fig:variable-gadget} and let $I'$ denote the bistable
gadget of $I$. We define the \emph{true (resp., false) assignment
of $I$} as the union of the true (resp., false) assignment of $I'$ and
$\{(b_1, \beta_1), (b_2, \beta_1), (a_4, \beta_2)\}$ (resp., $\{(a_2,
\beta_1), (b_3, \beta_2), (b_4, \beta_2)\}$).

\begin{lemma}
\label{lem:variable}
  The true and false assignments of $I$ are stable, and $I$ does not
  admit any other stable assignments.
\end{lemma}
\begin{deferproof}{Lemma~\ref{lem:variable}}
  Let $\Match_1$ (resp., $\Match_2$) denote the true (resp., false)
  assignment of $I$, let $B_1$ (resp., $B_2$) denote
  $\Block{I}{\Match_1}$ (resp., $\Block{I}{\Match_2}$), let $H$ denote
  $\Hosps{I'}$, let $\Match_1'$ (resp., $\Match_2'$) denote
  $\Proj{\Match_1}{H}$ (resp., $\Proj{\Match_2}{H}$), and let $B_1'$
  (resp., $B_2'$) denote $\Block{I'}{\Match_1'}$ (resp.,
  $\Block{I'}{\Match_2'}$). Thus $\Match_1'$ (resp., $\Match_2'$) is
  the true (resp., false) assignment of $I'$ and $B_1' = B_2' =
  \emptyset$.

  We begin by proving that $\Match_1$ is stable; a symmetric argument
  establishes that $\Match_2$ is stable. Since $\Match_1'$ is
  $H$-stable, Lemma~\ref{lem:outward} implies $\Match_1$ is
  $H$-stable. Since hospital $\beta_2$ is assigned to its most
  preferred resident $a_4$ and is full under $\Match_1$, we conclude
  that $\Match_1$ is $\{\beta_2\}$-stable. Since residents $a_2$,
  $b_1$, and $b_2$ are each assigned to their most preferred hospital
  in $\Match_1$, no blocking coalition in $B_1$ involves $a_2$, $b_1$,
  or $b_2$. It follows that $\Match_1$ is $\{\beta_1\}$-stable. Since
  $\Match_1$ is $(H \cup \{\beta_1, \beta_2\})$-stable, we conclude
  that $\Match_1$ is stable for $I$.

  It remains to prove that $\Match_1$ and $\Match_2$ are the only
  stable assignments for $I$. Let $\Match$ be a stable assignment for
  $I$, let $B$ denote $\Block{I}{\Match} = \emptyset$, and let
  $\Match'$ denote $\Proj{\Match}{H}$. Since $\Bdry{I}{H}{\Match} =
  \emptyset$ and $\Match$ is $H$-stable, Lemma~\ref{lem:inward}
  implies $\Match'$ is $H$-stable. Hence $\Match'$ is a stable
  assignment for $I'$. By Lemma~\ref{lem:bistable}, there are two
  cases to consider.

  Case~1: $\Match' = \Match_1'$. Since $(a_2, \alpha_1)$
  belongs to $\Match_1'$, the Case~1 condition implies $(a_2,
  \beta_1)$ does not belong to $\Match$. Since $(b_1, \beta_1)$
  (resp., $(b_2, \beta_1)$) does not belong to $B$, we deduce that
  $(b_1, \beta_1)$ (resp., $(b_2, \beta_1)$) belongs to
  $\Match$. Since $(a_4, \alpha_3)$ does not belong to $\Match_1'$,
  the Case~1 condition implies $(a_4, \alpha_3)$ does not belong to
  $\Match$. Since $(a_4, \beta_2)$ does not belong to $B$, we deduce
  that $(a_4, \beta_2)$ belongs to $\Match$. Thus residents $b_3$ and
  $b_4$ are unassigned in $\Match$. In summary, $\Match = \Match_1$.

  Case~2: $\Match' = \Match_2'$. By a symmetric argument to that used
  in Case~1, we find that $\Match = \Match_2$.
\end{deferproof}

\iffullversion
\else
\begin{figure}[t]
    \centering
    \begin{picture}(120.00,187.20)
\put(24.00,187.20){\makebox(0,0){\fbox{$\beta_1$}}}
\put(56.00,187.20){\makebox(0,0){\fbox{$\alpha_1\vphantom{\beta_1}$}}}
\put(88.00,187.20){\makebox(0,0){\fbox{$\alpha_3\vphantom{\beta_1}$}}}
\put(120.00,187.20){\makebox(0,0){\fbox{$\beta_2$}}}
\put(0.00,3.20){\makebox(0,0){$b_4$}}
\put(0.00,35.20){\makebox(0,0){$b_3$}}
\put(0.00,67.20){\makebox(0,0){\fbox{$a_4$}}}
\put(0.00,99.20){\makebox(0,0){\fbox{$a_2$}}}
\put(0.00,131.20){\makebox(0,0){$b_2$}}
\put(0.00,163.20){\makebox(0,0){$b_1$}}
\put(56.00,59.20){\line(1,0){32.00}}
\put(56.00,107.20){\line(1,0){32.00}}
\put(56.00,59.20){\line(0,1){48.00}}
\put(88.00,59.20){\line(0,1){48.00}}
\put(56.00,99.20){\circle*{6.40}}
\put(88.00,67.20){\circle*{6.40}}
\put(28.80,99.20){\vector(1,0){22.40}}
\put(24.00,99.20){\circle*{6.40}}
\put(24.00,126.40){\vector(0,-1){22.40}}
\put(24.00,131.20){\circle*{6.40}}
\put(24.00,158.40){\vector(0,-1){22.40}}
\put(24.00,163.20){\circle*{6.40}}
\put(115.20,67.20){\vector(-1,0){22.40}}
\put(120.00,67.20){\circle*{6.40}}
\put(120.00,40.00){\vector(0,1){22.40}}
\put(120.00,35.20){\circle*{6.40}}
\put(120.00,8.00){\vector(0,1){22.40}}
\put(120.00,3.20){\circle*{6.40}}
\end{picture}
    \caption{The preference digraph for the variable gadget includes a
      copy of the bistable gadget and has a total of $4 + 5 = 9$
      residents and $2 + 3 = 5$ hospitals.}
    \label{fig:variable-gadget}    
\end{figure}
\fi

\subsection{Clause Gadget}
\label{sec:clause}

In this section, we describe a clause gadget that contains three
copies of the literal gadget. The three copies of the literal gadget
have pairwise disjoint sets of agents (7 residents and 5 hospitals
each); we use a superscript $\ell$ in $\{1,2,3\}$ to identify the
agents in copy $\ell$. For example, we write $t_1^\ell$ to refer to
the resident $t_1$ in copy $\ell$. See Figure~\ref{fig:clause-gadget}
for a depiction of the clause gadget.

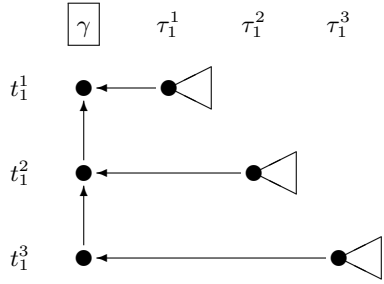
\begin{figure}
    \centering
    \begin{picture}(132.80,96.00)
\put(24.00,96.00){\makebox(0,0){\fbox{$\gamma\vphantom{\tau_1^1}$}}}
\put(56.00,96.00){\makebox(0,0){$\tau_1^1$}}
\put(88.00,96.00){\makebox(0,0){$\tau_1^2$}}
\put(120.00,96.00){\makebox(0,0){$\tau_1^3$}}
\put(0.00,8.00){\makebox(0,0){$t_1^3$}}
\put(0.00,40.00){\makebox(0,0){$t_1^2$}}
\put(0.00,72.00){\makebox(0,0){$t_1^1$}}
\put(24.00,8.00){\circle*{6.40}}
\put(120.00,8.00){\circle*{6.40}}
\put(120.00,8.00){\line(2,1){16.00}}
\put(120.00,8.00){\line(2,-1){16.00}}
\put(136.00,0.00){\line(0,1){16.00}}
\put(115.20,8.00){\vector(-1,0){86.40}}
\put(24.00,12.80){\vector(0,1){22.40}}
\put(24.00,40.00){\circle*{6.40}}
\put(88.00,40.00){\circle*{6.40}}
\put(88.00,40.00){\line(2,1){16.00}}
\put(88.00,40.00){\line(2,-1){16.00}}
\put(104.00,32.00){\line(0,1){16.00}}
\put(83.20,40.00){\vector(-1,0){54.40}}
\put(24.00,44.80){\vector(0,1){22.40}}
\put(24.00,72.00){\circle*{6.40}}
\put(56.00,72.00){\circle*{6.40}}
\put(56.00,72.00){\line(2,1){16.00}}
\put(56.00,72.00){\line(2,-1){16.00}}
\put(72.00,64.00){\line(0,1){16.00}}
\put(51.20,72.00){\vector(-1,0){22.40}}
\end{picture}
    \caption{The preference digraph for the clause gadget includes
      three copies of the literal gadget and has $3 \times 7 = 21$
      residents and $1 + 3 \times 5 = 16$ hospitals.}
    \label{fig:clause-gadget}
\end{figure}

Throughout the remainder of Section~\ref{sec:clause}, let $I$ denote
the clause gadget of Figure~\ref{fig:clause-gadget} and for any $\ell$
in $\{1,2,3\}$, let $I^\ell$ denote literal gadget $\ell$ of $I$ and
let $H^\ell$ denote $\Hosps{I^\ell}$. We say that an assignment
$\Match$ of $I$ is critical if for all $\ell$ in $\{1,2,3\}$,
$\Proj{\Match}{H^\ell}$ is the critical assignment of $I^\ell$. We say
that an assignment $\Match$ of $I$ is nice if every blocking coalition
in $\Block{I}{\Match}$ either involves hospital $\gamma$ or belongs to
$\{(t_1^1, \tau_1^1), (t_1^2, \tau_1^2), (t_1^3, \tau_1^3)\}$.

\begin{lemma}
  \label{lem:if-clause}
  Let $\Match$ be a critical assignment of $I$, let $\ell$ belong to
  $\{1,2,3\}$, let $B$ denote $\Block{I}{\Match}$, and let $B'$ denote
  $\Proj{B}{H^\ell}$. Then $B' \subseteq \{(t_1^\ell, \tau_1^\ell)\}$.
\end{lemma}
\begin{deferproof}{Lemma~\ref{lem:if-clause}}
  Let $\Match'$ denote $\Proj{\Match}{H^\ell}$ and let $B''$ denote
  $\Block{I^\ell}{\Match'}$. Since $\Match$ is critical, $\Match'$ is
  the critical assignment of literal gadget $I^\ell$. Hence
  Lemma~\ref{lem:literal} implies $B'' = \{(t^\ell_1,
  \tau^\ell_1)\}$. By Lemma~\ref{lem:outward}, we conclude that $B'
  \subseteq \{(t_1^\ell, \tau_1^\ell)\}$.
\end{deferproof}

\begin{lemma}
 \label{lem:only-if-clause}
 Let $\Match$ be a nice assignment of $I$, let $\ell$ belong
 to $\{1,2,3\}$, and let $\Match'$ denote
 $\Proj{\Match}{H^\ell}$. Then $\Match'$ is the critical assignment of
 $I^\ell$.
\end{lemma}
\begin{deferproof}{Lemma~\ref{lem:only-if-clause}}
  Let $B$ denote $\Block{I}{\Match}$, let $B'$ denote
  $\Proj{B}{H^\ell}$, let $B''$ denote $\Block{I^\ell}{\Match'}$, and
  assume for the sake of contradiction that $\Match'$ is not the
  critical assignment of $I^\ell$. Lemma~\ref{lem:literal} implies
  $B''$ includes a blocking coalition $(R,h) \neq (t_1^\ell,
  \tau_1^\ell)$ such that $h$ belongs to $H^\ell$. Since
  $\Bdry{I^\ell}{H^\ell}{\Match'} = \{(t_1^\ell, \tau_1^\ell)\}$ and
  $(R,h) \neq (t_1^\ell, \tau_1^\ell)$, Lemma~\ref{lem:inward} implies
  $(R,h)$ belongs to $B'$ and hence also to $B$. Since $h$ belongs to
  $H^\ell$, we deduce that $\Match$ is not a nice assignment of $I$, a
  contradiction.
\end{deferproof}

\begin{lemma}
  \label{lem:clause}
  Let $\Match$ be an assignment of $I$. Then $\Match$ is a nice
  assignment of $I$ if and only if $\Match$ is a critical assignment
  of $I$.
\end{lemma}
\begin{deferproof}{Lemma~\ref{lem:clause}}
  We first address the if direction. Assume that $\Match$ is a
  critical assignment of $I$. Let $\ell$ belong to $\{1,2,3\}$, let
  $B$ denote $\Block{I}{\Match}$, and let $B'$ denote
  $\Proj{B}{H^\ell}$. Lemma~\ref{lem:if-clause} implies $B' \subseteq
  \{(t_1^\ell, \tau_1^\ell)\}$. It follows that $\Match$ is a nice
  assignment of $I$.

  We now address the only if direction. Assume that $\Match$ is a nice
  assignment of $I$. Lemma~\ref{lem:only-if-clause} implies that for
  any $\ell$ in $\{1,2,3\}$, $\Proj{\Match}{H^\ell}$ is the critical
  assignment of $I^\ell$. Thus $\Match$ is a critical assignment of
  $I$, as required.
\end{deferproof}

\subsection{Transformation}
\label{sec:transform}

\iffullversion
\begin{figure}[h]
    \centering
    \begin{picture}(260.80,392.00)
\put(88.00,392.00){\makebox(0,0){\fbox{$\gamma^1$}}}
\put(120.00,392.00){\makebox(0,0){\fbox{$\gamma^2$}}}
\put(152.00,392.00){\makebox(0,0){\fbox{$\gamma^3$}}}
\put(184.00,392.00){\makebox(0,0){\fbox{$\gamma^4$}}}
\put(0.00,8.00){\makebox(0,0){$b_4^3$}}
\put(0.00,40.00){\makebox(0,0){$b_3^3$}}
\put(0.00,72.00){\makebox(0,0){$b_2^3$}}
\put(0.00,104.00){\makebox(0,0){$b_1^3$}}
\put(0.00,136.00){\makebox(0,0){$b_4^2$}}
\put(0.00,168.00){\makebox(0,0){$b_3^2$}}
\put(0.00,200.00){\makebox(0,0){$b_2^2$}}
\put(0.00,232.00){\makebox(0,0){$b_1^2$}}
\put(0.00,264.00){\makebox(0,0){$b_4^1$}}
\put(0.00,296.00){\makebox(0,0){$b_3^1$}}
\put(0.00,328.00){\makebox(0,0){$b_2^1$}}
\put(0.00,360.00){\makebox(0,0){$b_1^1$}}
\put(88.00,44.80){\vector(0,1){182.40}}
\put(88.00,236.80){\vector(0,1){118.40}}
\put(120.00,108.80){\vector(0,1){54.40}}
\put(120.00,172.80){\vector(0,1){118.40}}
\put(152.00,76.80){\vector(0,1){54.40}}
\put(152.00,140.80){\vector(0,1){182.40}}
\put(184.00,12.80){\vector(0,1){182.40}}
\put(184.00,204.80){\vector(0,1){54.40}}
\put(24.00,0.00){\line(1,0){32.00}}
\put(24.00,112.00){\line(1,0){32.00}}
\put(24.00,0.00){\line(0,1){112.00}}
\put(56.00,0.00){\line(0,1){112.00}}
\put(24.00,104.00){\circle*{6.40}}
\put(24.00,72.00){\circle*{6.40}}
\put(56.00,40.00){\circle*{6.40}}
\put(56.00,8.00){\circle*{6.40}}
\put(216.00,8.00){\circle*{6.40}}
\put(216.00,8.00){\line(2,1){16.00}}
\put(216.00,8.00){\line(2,-1){16.00}}
\put(232.00,0.00){\line(0,1){16.00}}
\put(211.20,8.00){\vector(-1,0){22.40}}
\put(179.20,8.00){\vector(-1,0){118.40}}
\put(184.00,8.00){\circle*{6.40}}
\put(216.00,40.00){\circle*{6.40}}
\put(216.00,40.00){\line(2,1){16.00}}
\put(216.00,40.00){\line(2,-1){16.00}}
\put(232.00,32.00){\line(0,1){16.00}}
\put(211.20,40.00){\vector(-1,0){118.40}}
\put(83.20,40.00){\vector(-1,0){22.40}}
\put(88.00,40.00){\circle*{6.40}}
\put(216.00,72.00){\circle*{6.40}}
\put(216.00,72.00){\line(2,1){16.00}}
\put(216.00,72.00){\line(2,-1){16.00}}
\put(232.00,64.00){\line(0,1){16.00}}
\put(211.20,72.00){\vector(-1,0){54.40}}
\put(147.20,72.00){\vector(-1,0){118.40}}
\put(152.00,72.00){\circle*{6.40}}
\put(216.00,104.00){\circle*{6.40}}
\put(216.00,104.00){\line(2,1){16.00}}
\put(216.00,104.00){\line(2,-1){16.00}}
\put(232.00,96.00){\line(0,1){16.00}}
\put(211.20,104.00){\vector(-1,0){86.40}}
\put(115.20,104.00){\vector(-1,0){86.40}}
\put(120.00,104.00){\circle*{6.40}}
\put(24.00,128.00){\line(1,0){32.00}}
\put(24.00,240.00){\line(1,0){32.00}}
\put(24.00,128.00){\line(0,1){112.00}}
\put(56.00,128.00){\line(0,1){112.00}}
\put(24.00,232.00){\circle*{6.40}}
\put(24.00,200.00){\circle*{6.40}}
\put(56.00,168.00){\circle*{6.40}}
\put(56.00,136.00){\circle*{6.40}}
\put(216.00,136.00){\circle*{6.40}}
\put(216.00,136.00){\line(2,1){16.00}}
\put(216.00,136.00){\line(2,-1){16.00}}
\put(232.00,128.00){\line(0,1){16.00}}
\put(211.20,136.00){\vector(-1,0){54.40}}
\put(147.20,136.00){\vector(-1,0){86.40}}
\put(152.00,136.00){\circle*{6.40}}
\put(216.00,168.00){\circle*{6.40}}
\put(216.00,168.00){\line(2,1){16.00}}
\put(216.00,168.00){\line(2,-1){16.00}}
\put(232.00,160.00){\line(0,1){16.00}}
\put(211.20,168.00){\vector(-1,0){86.40}}
\put(115.20,168.00){\vector(-1,0){54.40}}
\put(120.00,168.00){\circle*{6.40}}
\put(216.00,200.00){\circle*{6.40}}
\put(216.00,200.00){\line(2,1){16.00}}
\put(216.00,200.00){\line(2,-1){16.00}}
\put(232.00,192.00){\line(0,1){16.00}}
\put(211.20,200.00){\vector(-1,0){22.40}}
\put(179.20,200.00){\vector(-1,0){150.40}}
\put(184.00,200.00){\circle*{6.40}}
\put(216.00,232.00){\circle*{6.40}}
\put(216.00,232.00){\line(2,1){16.00}}
\put(216.00,232.00){\line(2,-1){16.00}}
\put(232.00,224.00){\line(0,1){16.00}}
\put(211.20,232.00){\vector(-1,0){118.40}}
\put(83.20,232.00){\vector(-1,0){54.40}}
\put(88.00,232.00){\circle*{6.40}}
\put(24.00,256.00){\line(1,0){32.00}}
\put(24.00,368.00){\line(1,0){32.00}}
\put(24.00,256.00){\line(0,1){112.00}}
\put(56.00,256.00){\line(0,1){112.00}}
\put(24.00,360.00){\circle*{6.40}}
\put(24.00,328.00){\circle*{6.40}}
\put(56.00,296.00){\circle*{6.40}}
\put(56.00,264.00){\circle*{6.40}}
\put(216.00,264.00){\circle*{6.40}}
\put(216.00,264.00){\line(2,1){16.00}}
\put(216.00,264.00){\line(2,-1){16.00}}
\put(232.00,256.00){\line(0,1){16.00}}
\put(211.20,264.00){\vector(-1,0){22.40}}
\put(179.20,264.00){\vector(-1,0){118.40}}
\put(184.00,264.00){\circle*{6.40}}
\put(216.00,296.00){\circle*{6.40}}
\put(216.00,296.00){\line(2,1){16.00}}
\put(216.00,296.00){\line(2,-1){16.00}}
\put(232.00,288.00){\line(0,1){16.00}}
\put(211.20,296.00){\vector(-1,0){86.40}}
\put(115.20,296.00){\vector(-1,0){54.40}}
\put(120.00,296.00){\circle*{6.40}}
\put(216.00,328.00){\circle*{6.40}}
\put(216.00,328.00){\line(2,1){16.00}}
\put(216.00,328.00){\line(2,-1){16.00}}
\put(232.00,320.00){\line(0,1){16.00}}
\put(211.20,328.00){\vector(-1,0){54.40}}
\put(147.20,328.00){\vector(-1,0){118.40}}
\put(152.00,328.00){\circle*{6.40}}
\put(216.00,360.00){\circle*{6.40}}
\put(216.00,360.00){\line(2,1){16.00}}
\put(216.00,360.00){\line(2,-1){16.00}}
\put(232.00,352.00){\line(0,1){16.00}}
\put(211.20,360.00){\vector(-1,0){118.40}}
\put(83.20,360.00){\vector(-1,0){54.40}}
\put(88.00,360.00){\circle*{6.40}}
\end{picture}
    \caption{The preference digraph for the $\HRSC$ instance
      corresponding to the $\SAT$ formula $(x_1 \lor x_2 \lor
      \neg{x_3}) \land (\neg{x_1} \lor \neg{x_2} \lor x_3) \land (x_1
      \lor \neg{x_2} \lor x_3) \land (\neg{x_1} \lor x_2 \lor
      \neg{x_3})$. Notice that while the dots associated with the
      literal icons all lie in the same column, there are no directed
      edges within this column. This is because each dot in this
      column is associated with a different copy of the literal gadget
      hospital $\tau_1$. For example, the topmost three dots in this
      column are associated with hospitals $\tau_1^{1,1}$,
      $\tau_1^{3,1}$, and $\tau_1^{2,1}$, respectively.}
    \label{fig:global-construction}
\end{figure}
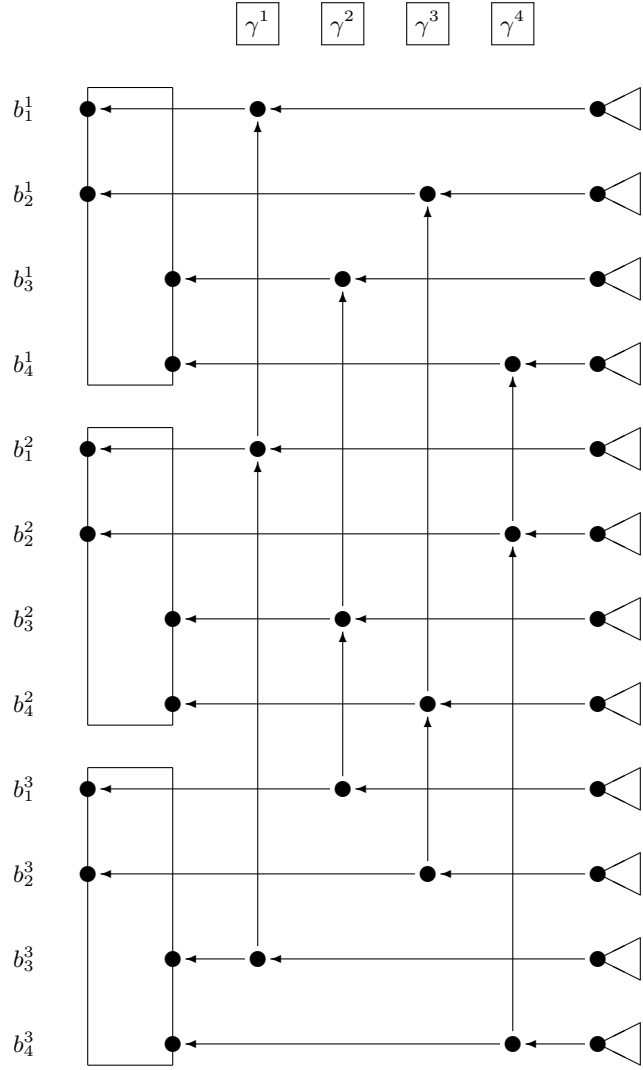
\fi

Throughout this section, we fix a $\SAT$ instance $\Prop$. Below we
describe how to transform $\Prop$ into an instance $I$ of $\HRSC$ so
that $I$ admits a stable assignment if and only if $\Prop$ is
satisfiable. It will be evident that the transformation is
polynomial-time computable.

Let $n$ be a positive integer such that $\Prop$ has $4n$ clauses and
$3n$ variables. We index the clauses (resp., variables) from 1 to $4n$
(resp., $3n$). We use two schemes for identifying the $12n$ literals
of $\Prop$. For any $i$ in $\{1,\dots, 4n\}$ and any $\ell$ in
$\{1,2,3\}$, we say that the literal in position $\ell$ of clause $i$
has A-identifier $(i,\ell)$. For any $j$ in $\{1,\dots,3n\}$, we say
that a literal has B-identifier $(j,1)$ (resp., $(j,2)$, $(j,3)$,
$(j,4)$) if it is the first positive (resp., second positive, first
negative, second negative) occurrence of variable $j$ in $\Prop$. In
order to translate between these two schemes, we define the bijection
$\BijFn: \{1,\dots,4n\} \times \{1,2,3\} \rightarrow \{1,\dots,3n\}
\times \{1,2,3,4\}$ to map any given A-identifier to the corresponding
B-identifier. Our instance $I$ will be constructed from $3n$ variable
gadgets and $4n$ clause gadgets. Below we discuss these components in
greater detail.

As discussed in Section~\ref{sec:variable}, our variable gadget has
$9$ residents and $5$ hospitals. We create $3n$ copies of this gadget,
one for each variable in $\Prop$; each copy has its own set of $9$
residents and $5$ hospitals. For any $j$ in $\{1,\ldots,3n\}$, we
define $V^j$ as copy $j$ of the variable gadget. We use a superscript
$j$ to identify the agents in $V^j$. For example, we write $a_1^j$ to
refer to the resident corresponding to $a_1$ in $V^j$. In our overall
instance $I$, the preference list of each hospital in $V^j$ is the
same as in the isolated gadget $V^j$; likewise, the preference list of
each resident in $V^j$ that does not belong to
$\{b^j_1,\ldots,b^j_4\}$ is the same as in the isolated gadget
$V^j$. Each resident in $\{b^j_1, \ldots, b^j_4\}$ has a preference
list of length 1 in $V^j$ and of length 3 in $I$: the most preferred
hospital of $b^j_k$ in $I$ is its lone acceptable hospital in $V^j$;
the two remaining hospitals on the preference list of $b^j_k$ do not
belong to $V^j$ and are specified below.

As discussed in Section~\ref{sec:clause}, our clause gadget has $21$
residents and $16$ hospitals. We create $4n$ copies of this gadget,
one for each clause in $\Prop$; each copy has its own set of $21$
residents and $16$ hospitals. For any $i$ in $\{1,\ldots,4n\}$, we
define $C^i$ as copy $i$ of the clause gadget. We use a superscript
$i$ to identify the agents in $C^i$. For example, we write $t_1^{i,1}$
to refer to the resident corresponding to $t_1^1$ in $C^i$. In our
overall instance $I$, the preference list of each hospital in $C^i$ is
the same as in the isolated gadget $C^i$; likewise, the preference
list of each resident in $C^i$ that does not belong to $\{t_1^{i,1},
t_1^{i,2}, t_1^{i,3}\}$ is the same as in the isolated gadget $C^i$.

It remains to describe the preference lists of the residents $b^j_k$
and $t_1^{i,\ell}$ in $I$. For each A-identifier $(i,\ell)$, we
identify resident $t_1^{i,\ell}$ with the resident $b^j_k$ such that
B-identifier $(j,k)$ is equal to $\Bij{i}{\ell}$. The preference list
of length 3 of resident $t_1^{i,\ell} = b^j_k$ in $I$ is obtained by
merging the preference list of $b^j_k$ in $V^j$ with the preference
list of $t_1^{i,\ell}$ in $C^i$, as follows: the top preference is the
lone acceptable hospital of $b^j_k$ in $V^j$; the second and third
preferences are the first and second preferences of $t_1^{i,\ell}$ in
$C^i$.

The above completes our description of the $\HRSC$ instance $I$
corresponding to $\SAT$ instance $\Prop$. We remark that the total
number of residents in $I$ is $3n \cdot 9 + 4n \cdot (21 - 3) = 99n$
and the total number of hospitals in $I$ is $3n \cdot 5 + 4n \cdot 16
= 79n$. The resident sizes and hospital capacities all belong to
$\{1,2\}$, and the longest preference list of any resident (resp.,
hospital) in $I$ is 3 (resp., 4).
\iffullversion
Figure~\ref{fig:global-construction}
gives an example of the transformation for a specific $\SAT$ instance
$\Prop$ with $n = 1$ (i.e., three variables and four clauses).
\fi

In the remainder of this section, we establish the correctness of the
polynomial-time transformation described above. That is, we prove that
$I$ admits a stable assignment if and only if $\Prop$ is satisfiable.

\subsubsection{If Direction}
\label{sec:if}

Assume that $\Prop$ is satisfiable, and fix a satisfying truth
assignment $\pi=(\pi_1, \ldots, \pi_{3n})$ for $\Prop$; in other
words, $\pi_j$ is the truth value of variable $j$ under $\pi$. We need
to prove that $\HRSC$ instance $I$ admits a stable assignment.

We now describe how to use $\pi$ to construct an assignment $\Match$
for $I$. After describing the construction, we prove that $\Match$ is
a valid assignment for $I$ (i.e., no resident or hospital is
over-subscribed), and then we prove that $\Match$ is stable for
$I$.

For each variable $j$ such that $\pi_j$ is true (resp., false), we
include the true (resp., false) assignment of $V^j$ in $\Match$. For
each clause $i$, we include in $\Match$ the critical assignment of
$C^i$ that, for each $\ell$ in $\{1,2,3\}$, assigns resident
$t_1^{i,\ell}$ to hospital $\gamma^i$ if and only if the literal with
A-identifier $(i,\ell)$ evaluates to false under truth assignment
$\pi$. This completes the definition of assignment $\Match$.

We now argue that $\Match$ is a valid assignment of $I$. It is easy to
see that for any $j$ in $\{1,\ldots,3n\}$, no hospital in $V^j$ is
over-subscribed and that no resident in $V^j$ outside the set
$\{b_1^j,\ldots,b_4^j\}$ is over-subscribed. Similarly, it is easy to
see that for any $i$ in $\{1,\ldots,4n\}$, no hospital in $C^i$ apart
from $\gamma^i$ is over-subscribed and no resident in $C^i$ outside
the set $\{t_1^{i,1}, t_1^{i,2}, t_1^{i,3}\}$ is over-subscribed. Thus
Lemmas~\ref{lem:if-gamma} and~\ref{lem:if-identified} below imply
$\Match$ is a valid assignment of $I$.

\begin{lemma}
  \label{lem:if-gamma}
  Let $i$ belong to $\{1,\ldots,4n\}$. Then hospital $\gamma^i$ is not
  over-subscribed.
\end{lemma}
\begin{deferproof}{Lemma~\ref{lem:if-gamma}}
  Since $\pi$ is a satisfying truth assignment, $\gamma^i$ has
  capacity 2, the set of acceptable residents of $\gamma^i$ is
  $\{t_1^{i,\ell}\mid \ell \in \{1,2,3\}\}$, and each resident in this
  set has size 1, we deduce that $\gamma^i$ is not over-subscribed.
\end{deferproof}

\begin{lemma}
  \label{lem:if-top-two}
  Let $i$ belong to $\{1,\ldots,4n\}$ and let $\ell$ belong to
  $\{1,2,3\}$. Resident $t_1^{i,\ell}$ is assigned under $\Match$ to its
  most (resp., second most) preferred hospital if and only if literal
  $\ell$ of clause $i$ evaluates to true (resp., false) under $\pi$.
\end{lemma}
\begin{deferproof}{Lemma~\ref{lem:if-top-two}}
  Let $(j,k)$ denote the B-identifier $\Bij{i}{\ell}$. Observe that
  the true (resp., false) assignment for variable gadget $j$ assigns
  (resp., does not assign) residents $b^j_1$ and $b^j_2$ to
  $\beta^j_1$, their most preferred hospital in $I$, and does not
  assign (resp., assigns) residents $b^j_3$ and $b^j_4$ to
  $\beta^j_2$, their most preferred hospital in $I$. The claim of the
  lemma follows straightforwardly from the definition of $\Match$.
\end{deferproof}

\begin{lemma}
  \label{lem:if-identified}
 Let $i$ belong to $\{1,\ldots,4n\}$, let $\ell$ belong to
 $\{1,2,3\}$, and let $(j,k)$ denote the B-identifier
 $\Bij{i}{\ell}$. Then resident $t_1^{i,\ell} = b^j_k$ is not
 over-subscribed.
\end{lemma}
\begin{deferproof}{Lemma~\ref{lem:if-identified}}
  Since $\Proj{\Match}{\Hosps{C^i}}$ is a critical assignment of $C^i$,
  resident $t_1^{i,\ell}$ is not assigned to its third preference in
  $I$ (i.e., $\tau_1^{i,\ell}$) under $\Match$. By
  Lemma~\ref{lem:if-top-two}, we conclude that resident $t_1^{i,\ell}$
  is not over-subscribed.
\end{deferproof}

It remains to prove that $\Match$ is stable for $I$.

\begin{lemma}
  \label{lem:if-clause-global}
  Let $i$ belong to $\{1,\dots,4n\}$ and let $H$ denote $\Hosps{C^i}$.
  Then $\Match$ is an $H$-stable assignment of $I$.
\end{lemma}
\begin{deferproof}{Lemma~\ref{lem:if-clause-global}}
  Let $B$ denote $\Block{I}{\Match}$. The construction of $\Match$
  implies $\Match' = \Proj{\Match}{H}$ is a critical assignment of
  $C^i$. Thus Lemma~\ref{lem:clause} implies $\Match'$ is a nice
  assignment of $C^i$. Hence Lemma~\ref{lem:outward} implies that for
  every blocking coalition $(R,h)$ in $B$ such that $h$ belongs to
  $\Hosps{C^i}$, either $h = \gamma^i$ or $(R,h)$ belongs to
  $$
  \{(t_1^{i,1}, \tau_1^{i,1}), (t_1^{i,2}, \tau_1^{i,2}), (t_1^{i,3},
  \tau_1^{i,3})\}.
  $$ Thus Claims~1 and~2 below imply $\Match$ is an $H$-stable
  assignment of $I$, as required.

  Claim~1: The hospital $\gamma^i$ is not involved in any blocking
  coalition in $B$. The only residents who could participate in a
  blocking coalition in $B$ with $\gamma^i$ are $t_1^{i,1}$,
  $t_1^{i,2}$, and $t_1^{i,3}$. Lemma~\ref{lem:if-top-two} implies
  that for any $\ell$ in $\{1,2,3\}$, resident $t_1^{i,\ell}$ is
  assigned to one of its two most preferred hospitals. Since
  $\gamma^i$ is the second preference of $t_1^{i,\ell}$, we deduce
  that $t_1^{i,\ell}$ does not participate in a blocking coalition in
  $B$ with $\gamma^i$. This completes the proof of Claim~1.

  Claim~2: For any $\ell$ in $\{1,2,3\}$, hospital $\tau_1^{i,\ell}$
  is not involved in any blocking coalition in $B$. Let $\ell$ belong
  to $\{1,2,3\}$. Size and capacity considerations imply that the only
  possible blocking coalitions in $B$ involving $\tau_1^{i,\ell}$ are
  $(t_1^{i,\ell}, \tau_1^{i,\ell})$ and $(t_2^{i,\ell},
  \tau_1^{i,\ell})$. Lemma~\ref{lem:if-top-two} implies resident
  $t_1^{i,\ell}$ is assigned to one of its two most preferred
  hospitals. Since $\tau_1^{i,\ell}$ is its third choice,
  $(t_1^{i,\ell}, \tau_1^{i,\ell})$ does not belong to $B$. Since
  $\Match$ assigns $t_2^{i,\ell}$ to $\tau_1^{i,\ell}$,
  $(t_2^{i,\ell}, \tau_1^{i,\ell})$ does not belong to $B$. This
  completes the proof of Claim~2.
\end{deferproof}

\begin{lemma}
  \label{lem:if-variable}
  Let $j$ belong to $\{1,\ldots,3n\}$ and let $H$ denote
  $\Hosps{V^j}$. Then $\Match$ is an $H$-stable assignment of $I$.
\end{lemma}
\begin{deferproof}{Lemma~\ref{lem:if-variable}}
  Let $B$ denote $\Block{I}{\Match}$, let $B'$ denote $\Proj{B}{H}$,
  let $\Match'$ denote $\Proj{\Match}{H}$, and let $B''$ denote
  $\Block{V^j}{\Match'}$. By the definition of $\Match$, $\Match'$ is
  the true or the false assignment for $V^j$. Thus
  Lemma~\ref{lem:variable} implies $\Match'$ is stable for $V^j$ and
  hence $B'' = \emptyset$. Hence Lemma~\ref{lem:outward} implies $B' =
  \emptyset$; the claim of the lemma follows.
\end{deferproof}

\begin{lemma}
  \label{lem:if-stable}
  Assignment $\Match$ is stable for $I$.
\end{lemma}
\begin{deferproof}{Lemma~\ref{lem:if-stable}}
  Since Lemma~\ref{lem:if-clause-global} holds for all $i$ in
  $\{1,\dots,4n\}$ and Lemma~\ref{lem:if-variable} holds for all $j$
  in $\{1,\dots,3n\}$, we conclude that $\Match$ is a stable
  assignment for $I$.
\end{deferproof}

\subsubsection{Only If Direction}
\label{sec:only-if}

Assume that $I$ admits a stable assignment, and fix a stable
assignment $\Match$ of $I$. We need to prove that $\Prop$ is
satisfiable.

\begin{lemma}
  \label{lem:only-if-variable}
  Let $j$ belong to $\{1,\ldots,3n\}$, let $H$ denote $\Hosps{V^j}$,
  and let $\Match'$ denote $\Proj{\Match}{H}$. Then $\Match'$ is
  the true or the false assignment for $V^j$.
\end{lemma}
\begin{deferproof}{Lemma~\ref{lem:only-if-variable}}
  Since $\Bdry{I}{H}{\Match} = \emptyset$, Lemma~\ref{lem:inward}
  implies $\Block{V^j}{\Match'}$ is equal to
  $\Proj{\Block{I}{\Match}}{H} = \Proj{\emptyset}{H} =
  \emptyset$. Thus $\Match'$ is a stable assignment for $V^j$ and the
  claim of the lemma follows from Lemma~\ref{lem:variable}.
\end{deferproof}

Let $\pi = (\pi_1, \ldots, \pi_{3n})$ be the truth assignment for
$\Prop$ such that $\pi_j$ is true if and only if
$\Proj{\Match}{\Hosps{V^j}}$ is the true assignment for $V^j$. Below
we complete the ``only if'' proof by showing that $\pi$ is a
satisfying truth assignment for $\Prop$.

\begin{lemma}
  \label{lem:only-if-true}
  Let $i$ belong to $\{1,\ldots,4n\}$ and let $\ell$ belong to
  $\{1,2,3\}$. Then the literal of $\Prop$ with A-identifier
  $(i,\ell)$ evaluates to true under $\pi$ if and only if $\Match$
  assigns resident $t_1^{i,\ell}$ to its most preferred hospital.
\end{lemma}
\begin{deferproof}{Lemma~\ref{lem:only-if-true}}
  Let $(j,k)$ denote the B-identifier $\Bij{i}{\ell}$, let $H$ denote
  $\Hosps{V^j}$, and let $\Match'$ denote $\Proj{\Match}{H}$. The most
  preferred hospital of resident $t_1^{i,\ell} = b^j_k$ is
  $\beta^j_{\lceil k/2 \rceil}$. By Lemma~\ref{lem:only-if-variable},
  $\Match'$ is the true or the false assignment for $V^j$. If $k$
  belongs to $\{1,2\}$, then resident $t_1^{i,\ell}$ is (resp., is
  not) assigned to its most preferred hospital under the true (resp.,
  false) assignment for $V^j$. If $k$ belongs to $\{3,4\}$, then
  resident $t_1^{i,\ell}$ is (resp., is not) assigned to its most
  preferred hospital under the false (resp., true) assignment for
  $V^j$. The claim of the lemma follows.
\end{deferproof}

\begin{lemma}
  \label{lem:only-if-top-two}
  Let $i$ belong to $\{1,\ldots,4n\}$ and let $\ell$ belong to
  $\{1,2,3\}$. Then $\Match$ assigns resident $t_1^{i,\ell}$ to one of
  its two most preferred hospitals.
\end{lemma}
\begin{deferproof}{Lemma~\ref{lem:only-if-top-two}}
  Let $B$ denote $\Block{I}{\Match}$, let $H$ denote $\Hosps{C^i}$,
  let $B'$ denote $\Proj{B}{H}$, let $I'$ denote $\Proj{I}{H}$, let
  $\Match'$ denote $\Proj{\Match}{H}$, and let $B''$ denote
  $\Block{I'}{\Match'}$. Since $B = B' = \emptyset$ and every blocking
  coalition in $\Bdry{I}{H}{\Match}$ either involves hospital
  $\gamma^i$ or is contained in $\{(t_1^{i,1}, \tau_1^{i,1}),
  (t_1^{i,2}, \tau_1^{i,2}), (t_1^{i,3}, \tau_1^{i,3})\}$,
  Lemma~\ref{lem:inward} implies $\Match'$ is a nice assignment for
  $I'$. Thus Lemma~\ref{lem:clause} implies $\Match'$ is a critical
  assignment for $I'$. It follows that $\Match$ does not assign
  resident $t_1^{i,\ell}$ to its least preferred hospital
  $\tau_1^{i,\ell}$ (i.e., its third choice).

  It remains to argue that $t_1^{i,\ell}$ cannot be unassigned in
  $\Match$. Suppose $t_1^{i,\ell}$ is unassigned in $\Match$. Then
  $(t_1^{i,\ell}, \tau_1^{i,\ell})$ is a blocking coalition in $B$, a
  contradiction.
\end{deferproof}

\begin{lemma}
  \label{lem:only-if-false}
  Let $i$ belong to $\{1,\ldots,4n\}$ and let $\ell$ belong to
  $\{1,2,3\}$. Then the literal of $\Prop$ with A-identifier
  $(i,\ell)$ evaluates to false under truth assignment $\pi$ if and
  only if $\Match$ assigns $t_1^{i,\ell}$ to $\gamma^i$.
\end{lemma}
\begin{deferproof}{Lemma~\ref{lem:only-if-false}}
  Immediate from Lemmas~\ref{lem:only-if-true}
  and~\ref{lem:only-if-top-two}, since hospital $\gamma^i$ is the
  second most preferred hospital of resident $t_1^{i,\ell}$.
\end{deferproof}

\begin{lemma}
  \label{lem:only-if-main}
  Truth assignment $\pi$ satisfies $\Prop$.
\end{lemma}
\begin{deferproof}{Lemma~\ref{lem:only-if-main}}
  Let $i$ belong to $\{1,\ldots,4n\}$. We need to prove that
  truth assignment $\pi$ satisfies clause $i$ of $\Prop$. Assume for
  the sake of contradiction that $\pi$ does not satisfy clause $i$. By
  Lemma~\ref{lem:only-if-false}, $\Match$ assigns the three unit-size
  residents $t_1^{i,1}$, $t_1^{i,2}$, and $t_1^{i,3}$ to
  $\gamma^i$. Since the capacity of $\gamma^i$ is 2, we conclude that
  $\Match$ is not a valid assignment for $I$, a contradiction.
\end{deferproof}

\iffullversion
  \subsection{Complete Preferences}
\label{sec:complete}

Our proof that the decision version of $\HRSC$ is NP-hard makes use of
$\HRSC$ instances with incomplete preferences, that is, where a resident
(resp., hospital) only provides a preference order over its
acceptable hospitals (resp., residents).  In this section, we use
a simple reduction to show that the problem remains NP-hard in the
special case where all preference lists are complete.

Let $I$ denote an instance of $\HRSC$ with incomplete preferences and
where the resident sizes and hospital capacities all belong to
$\{1,2\}$.  For any resident $r$ in $I$, let $A_r$ denote the set of
hospitals that are acceptable to $r$.  Similarly, for any hospital
$h$ in $I$, let $A_h$ denote the set of residents that are acceptable
to $h$.

We now describe how to transform an arbitrary $\HRSC$ instance $I$
into an ``equivalent'' instance $I'$ of $\HRSC$ with complete
preferences and resident sizes and hospital capacities in
$\{1,2\}$. The set of hospitals in $I'$ is the same as in $I$. The set
of residents in $I'$ consists of those in $I$ plus the following
``dummy'' residents: For each capacity-$k$ hospital $h$ in $I$, we
introduce a set of $k$ size-$1$ dummy residents
$\{r_h^1,\ldots,r_h^k\}$. Each agent in instance $I'$ has complete
preferences, determined as follows.
\begin{itemize}
\item
For any resident $r$ in $I$, the preference list of $r$ in $I'$
consists of its preference list in $I$ (which contains exactly the
hospitals in $A_r$) followed by the remaining hospitals in arbitrary
order.

\item
For any dummy resident $r_h^i$, the associated preference list in
$I'$ consists of $h$ followed by all other hospitals in arbitrary
order.

\item
For any hospital $h$ in $I$, the preference list of $h$ in $I'$
consists of its preference list in $I$ (which contains exactly the
residents in $A_h$) followed by dummy residents $r_h^1$ through
$r_h^k$ (in that order), followed by the remaining (non-dummy and
dummy) residents in arbitrary order.
\end{itemize}

We define an assignment $\Match'$ in $I'$ to be \emph{special} if the
following conditions hold for any hospital $h$, where $R_h$ denotes
the set of all residents assigned to $h$ under $\Match'$: (1)
$\Size{R_h}=\Cap{h}$; (2) there is an $\ell$ in
$\{0,\ldots,\Cap{h}\}$ such that $R_h$ is the union of
$\{r_h^1,\ldots,r_h^\ell\}$ and some subset of $A_h$.

There is a natural bijection between the set of all assignments for
$I$ and the set of all special assignments for $I'$. Specifically,
given an assignment $\Match$ for $I$, we can obtain the corresponding
special assignment $\Match'$ for $I'$ by assigning, for each hospital $h$
with unused capacity $\ell$ under $\Match$, dummy residents
$r_h^1,\ldots,r_h^\ell$ to hospital $h$. And given a special assignment
$\Match'$ for $I'$, we can obtain the corresponding assignment $\Match$
for $I$ by removing any resident-hospital pairs involving dummy
residents.

The following two claims are straightforward to prove: (1) if $\Match'$
is stable for $I'$, then $\Match'$ is special; (2) if $\Match$ is an
assignment for $I$ and $\Match'$ is the corresponding special assignment for
$I'$, then $\Block{I}{\Match}=\Block{I'}{\Match'}$. It follows that $I$
admits a stable matching if and only if $I'$ admits a stable matching.

\fi

\section{Unitwise-Coalition Stability}
\label{sec:unitwise}

As discussed in Section~\ref{sec:intro}, Rodr\'{i}guez and Manlove
have considered two coalitional notions of stability in the context of
course allocation, namely first-coalition stability and coalition
stability, which correspond to $\HRSFC$ and $\HRSC$ in the
special-case setting of $\HRS$. Rodr\'{i}guez and Manlove have shown
that $\HRSFC$ is NP-hard and we have shown that $\HRSC$ is NP-hard. In
this section, we discuss a third notion of coalitional stability in
the context of $\HRS$ that is arguably more natural than either
first-coalition stability or coalition stability, and we indicate why
our NP-hardness result for $\HRSC$ extends to this third notion, which
we call unitwise-coalition stability.

To motivate this third coalitional notion of stability, consider the
following example. Suppose hospital $h$ has $\Cap{h} = 4$ and
preference list $r_1$ (most preferred), $r_2$, $r_3$, $r_4$ (least
preferred) where $\Size{r_1} = 2$, $\Size{r_2} = 1$, $\Size{r_3} = 2$,
and $\Size{r_4} = 3$. Further assume that assignment $\Match$ assigns
residents $r_2$ and $r_4$ to hospital $h$. In this scenario, it is
straightforward to verify that there is no blocking coalition (in the
sense of $\HRSC$) for $\Match$ involving hospital $h$. Under our
proposed coalitional notion of stability, we consider $(\{r_1,r_3\},
h)$ to be a blocking coalition because $h$ ``should'' prefer
$\{r_1,r_3\}$ to $\{r_2,r_4\}$: the two units of $r_1$ are preferred
to the single unit of $r_2$ and one unit of $r_4$; the two units of
$r_3$ are preferred to the two remaining units of $r_4$. More
generally, the idea is that a hospital $h$ prefers a set of residents
$R$ to a set of residents $R'$ (disjoint from $R$) if $\Size{R} \ge
\Size{R'}$ and for all $i$ in $\{1,\ldots, \Size{R'}\}$, $h$ prefers
its $i$th most preferred unit in $R$ to its $i$th most preferred unit
in $R'$.

Unitwise-coalition stability admits an equivalent definition in terms
of additive utility. Assume that the utility hospital $h$ derives from
being assigned a set of residents $R$ is equal to the sum of its
utility for each resident $r$ in $R$, where the utility for each
resident is a positive real number. Further assume that the utilities
respect the preference list of $h$, in the following sense: if $h$
prefers $r$ to $r'$, then the per-unit utility of $r$ (i.e., the
utility of $h$ for $r$ divided by $\Size{r}$) is greater than the
per-unit utility of $r'$ (i.e., the utility of $h$ for $r'$ divided by
$\Size{r'}$). We consider hospital $h$ to prefer a set of residents
$R$ to a set of residents $R'$ (disjoint from $R$) if $\Size{R} \ge
\Size{R'}$ and the utility of $h$ for $R$ exceeds the utility of $h$
for $R'$ under any utilities that respect the preference list of
$h$. It is straightforward to prove that this utility-based
formulation is equivalent to unitwise-coalition stability.

We now indicate why the NP-hardness proof of Section~\ref{sec:nph}
also applies to unitwise-coalition stability. Our proof establishes
that $\HRSC$ is NP-hard even in the special case where the resident
sizes and the hospital capacities belong to $\{1,2\}$. It is
straightforward to verify that in this special case, an assignment is
coalition-stable (i.e., stable under $\HRSC$) if and only if it is
unitwise-coalition-stable.
\iffullversion
\footnote{The set of blocking coalitions under unitwise-coalition
stability can be strictly larger than under coalition stability, but is
nonempty if and only if it is nonempty under coalition stability.}
\fi

\section{Concluding Remarks}
\label{sec:concluding-remarks}

In this paper, we presented an instance of $\HRIC$ that does not admit
a coalition-stable assignment. In addition, we established that it is
NP-hard to find a coalition-stable assignment for a given $\HRIC$
instance (or correctly report that no such assignment exists), even if
every hospital has capacity at most two and the length of the
preference list of each resident (resp., hospital) is at most three
(resp., four). As discussed in Section~\ref{sec:unitwise}, this
hardness result also holds for unitwise-coalition stability. Since
$\HRS$ generalizes $\HRIC$ and course allocation generalizes $\HRS$,
our results resolve the open questions of Rodr\'{i}guez and
Manlove~\cite{Rodriguez2025} related to finding a coalition-stable
assignment in the course allocation setting.

As remarked in Section~\ref{sec:intro}, Rodr\'{i}guez and Manlove
prove that the problem of testing whether a given assignment is
coalition-stable is co-NP-complete for course allocation, and their
argument also applies to the $\HRS$ special case. On the other hand,
we claim that the testing problem for $\HRIC$ under pair-size,
first-coalition, coalition, and unitwise-coalition stability belongs
to P. To show this, we first argue that if there is a blocking
coalition involving some hospital $h$, then there is a blocking
coalition $(R,h)$ such that $\Size{R} \le 2$. Furthermore, for any
hospital $h$ and any set of residents $R$ such that $\Size{R} \le 2$,
it is straightforward to check in polynomial time whether $(R,h)$ is a
blocking coalition. Since there are only polynomially many candidate
blocking coalitions $(R,h)$ with $\Size{R} \le 2$, the claim follows.

\bibliographystyle{splncs04}
\bibliography{references}

\iffullversion
\appendix
\renewcommand{\theHsection}{\Alph{section}}

\section{Tabular Description of the Transformation}
\label{sec:tabular-transformation}

Let $\Prop$ be a $\SAT$ instance, let $n$ be a positive integer such
that $\Prop$ has $4n$ clauses, indexed from 1 to $4n$, and $3n$
variables, indexed from 1 to $3n$. As in Section~\ref{sec:transform},
we define both an A-identifier and a B-identifier for each literal in
$\Prop$ and we define the bijection $\BijFn: \{1,\dots,4n\} \times
\{1,2,3\} \rightarrow \{1,\dots,3n\} \times \{1,2,3,4\}$ to map any
given A-identifier to the corresponding B-identifier.

\begin{table}[h]
\centering
\begin{tabular}{ccl@{\qquad}ccl}
\toprule
resident & size & preferences & hospital & capacity & preferences \\
\midrule
$a_1^j$             & $1$ & $\alpha_1^j \quad\alpha_3^j$                                     & $\alpha_1^j$      & $2$ & $a_5^j \quad a_3^j \quad a_2^j \quad a_1^j$ \\
$a_2^j$             & $2$ & $\alpha_1^j \quad\beta_1^j$                                       & $\alpha_2^j$      & $1$ & $a_3^j$ \\
$a_3^j$             & $1$ & $\alpha_1^j \quad\alpha_2^j$                                      & $\alpha_3^j$      & $2$ & $a_1^j \quad a_4^j \quad a_5^j$ \\
$a_4^j$             & $2$ & $\alpha_3^j \quad\beta_2^j$                                       & $\beta_1^j$       & $2$ & $a_2^j \quad b_2^j \quad b_1^j$ \\
$a_5^j$             & $1$ & $\alpha_3^j \quad\alpha_1^j$                                      & $\beta_2^j$       & $2$ & $a_4^j \quad b_3^j \quad b_4^j$ \\
$b_k^j = t_1^{i,\ell}$ & $1$ & $\beta_{\lceil k/2 \rceil}^j \quad\gamma^i \quad\tau_1^{i,\ell}$ & $\gamma^i$      & $2$ & $t_1^{i,1} \quad t_1^{i,2} \quad t_1^{i,3}$ \\
$s_1^{i,\ell}$      & $1$ & $\sigma_1^{i,\ell} \quad\sigma_2^{i,\ell}$                        & $\sigma_1^{i,\ell}$ & $1$ & $s_2^{i,\ell} \quad s_1^{i,\ell}$ \\
$s_2^{i,\ell}$      & $1$ & $\sigma_3^{i,\ell} \quad\sigma_2^{i,\ell} \quad\sigma_1^{i,\ell}$  & $\sigma_2^{i,\ell}$ & $2$ & $s_1^{i,\ell} \quad s_2^{i,\ell} \quad s_3^{i,\ell} \quad s_4^{i,\ell}$ \\
$s_3^{i,\ell}$      & $2$ & $\tau_2^{i,\ell} \quad\sigma_2^{i,\ell}$                          & $\sigma_3^{i,\ell}$ & $1$ & $s_4^{i,\ell} \quad s_2^{i,\ell}$ \\
$s_4^{i,\ell}$      & $1$ & $\sigma_2^{i,\ell} \quad\sigma_3^{i,\ell}$                        & $\tau_1^{i,\ell}$ & $1$ & $t_1^{i,\ell} \quad t_2^{i,\ell}$ \\
$t_2^{i,\ell}$      & $1$ & $\tau_1^{i,\ell} \quad\tau_2^{i,\ell}$                            & $\tau_2^{i,\ell}$ & $2$ & $t_2^{i,\ell} \quad t_3^{i,\ell} \quad s_3^{i,\ell}$ \\
$t_3^{i,\ell}$      & $1$ & $\tau_2^{i,\ell}$                                               & & & \\
\bottomrule \\
\end{tabular}
\caption{A tabular representation of the construction of
  Section~\ref{sec:transform} (see also
  Figure~\ref{fig:global-construction}). For each resident (resp.,
  hospital), the acceptable hospitals (resp., residents) are listed in
  decreasing order of preference.}
\label{tab:transformation}
\end{table}

\section{Preference Digraphs for the Gadgets Implicit
  in~\cite{McDermid2010}}
\label{sec:hrs-gadgets}

The NP-hardness proof for $\HRSP$ given
in~\cite[erratum]{McDermid2010} is based on a tabular representation
of the construction (see Figure 1 of \cite[erratum]{McDermid2010})
that is analogous to Table~\ref{tab:transformation} of
Appendix~\ref{sec:tabular-transformation} for our construction. In
this section, we present preference digraphs for the gadgets implicit
in the construction of~\cite{McDermid2010}, to allow the reader to
better understand the relationship between our NP-hardness proof for
$\HRSC$ and their NP-hardness proof for $\HRSP$. We do not provide a
preference digraph for the clause gadget implicit
in~\cite{McDermid2010} because it is the same as ours (see
Figure~\ref{fig:clause-gadget}), with the understanding that the
literal gadget icons correspond to Figure~\ref{fig:hrs-literal-gadget}
rather than Figure~\ref{fig:literal-gadget}. Similarly, we do not
provide the preference digraph for the global construction
of~\cite{McDermid2010} because it is the same as ours (see
Figure~\ref{fig:global-construction}), with the understanding that the
variable and literal gadget icons correspond to
Figures~\ref{fig:variable-gadget} and~\ref{fig:literal-gadget} rather
than Figures~\ref{fig:hrs-variable-gadget}
and~\ref{fig:hrs-literal-gadget}.

Figure~\ref{fig:hrs-variable-gadget} depicts the preference digraph
for the variable gadget implicit in~\cite{McDermid2010}. This
construction can be streamlined by instead applying the construction
of Figure~\ref{fig:variable-gadget}, with the understanding that the
bistable gadget icon corresponds to
Figure~\ref{fig:hrs-bistable-gadget} rather than
Figure~\ref{fig:bistable-gadget}.

The symbols used to denote specific residents and hospitals in
Figures~\ref{fig:hrs-unstable-gadget}
through~\ref{fig:hrs-variable-gadget} are based on the notation
in~\cite{McDermid2010}; the reader is referred to~\cite{McDermid2010}
for further details.

\begin{figure}
    \centering
    \begin{picture}(56.00,91.20)
\put(24.00,91.20){\makebox(0,0){\fbox{$p^k_{j,1}$}}}
\put(56.00,91.20){\makebox(0,0){$p^k_{j,2}$}}
\put(0.00,3.20){\makebox(0,0){$q^k_{j,2}$}}
\put(0.00,35.20){\makebox(0,0){\fbox{$q^k_{j,3}$}}}
\put(0.00,67.20){\makebox(0,0){$q^k_{j,1}$}}
\put(24.00,3.20){\circle*{6.40}}
\put(24.00,8.00){\vector(0,1){22.40}}
\put(24.00,35.20){\circle*{6.40}}
\put(24.00,40.00){\vector(0,1){22.40}}
\put(24.00,67.20){\circle*{6.40}}
\put(28.80,67.20){\vector(1,0){22.40}}
\put(56.00,67.20){\circle*{6.40}}
\put(56.00,62.40){\vector(0,-1){54.40}}
\put(56.00,3.20){\circle*{6.40}}
\put(51.20,3.20){\vector(-1,0){22.40}}
\end{picture}
    \caption{The preference digraph for copy $(j,k)$ of the unstable
      gadget implicit in~\cite{McDermid2010} has three residents
      $q_{j,1}^k, q_{j,2}^k, q_{j,3}^k$ and two hospitals $p_{j,1}^k,
      p_{j,2}^k$.}
    \label{fig:hrs-unstable-gadget}
\end{figure}
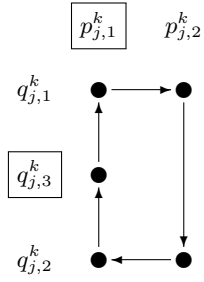

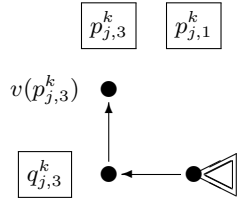
\begin{figure}
    \centering
    \begin{picture}(68.80,64.00)
\put(24.00,64.00){\makebox(0,0){\fbox{$p^k_{j,3}$}}}
\put(56.00,64.00){\makebox(0,0){\fbox{$p^k_{j,1}$}}}
\put(0.00,8.00){\makebox(0,0){\fbox{$q^k_{j,3}$}}}
\put(0.00,40.00){\makebox(0,0){$v(p^k_{j,3})$}}
\put(56.00,8.00){\circle*{6.40}}
\put(56.00,8.00){\line(2,1){16.00}}
\put(56.00,8.00){\line(2,-1){16.00}}
\put(72.00,0.00){\line(0,1){16.00}}
\put(60.00,8.00){\line(2,1){10.56}}
\put(60.00,8.00){\line(2,-1){10.56}}
\put(70.56,2.72){\line(0,1){10.56}}
\put(51.20,8.00){\vector(-1,0){22.40}}
\put(24.00,8.00){\circle*{6.40}}
\put(24.00,12.80){\vector(0,1){22.40}}
\put(24.00,40.00){\circle*{6.40}}
\end{picture}
    \caption{The preference digraph for copy $(j,k)$ of the literal
      gadget implicit in~\cite{McDermid2010} includes an icon
      representing a copy of the unstable gadget of
      Figure~\ref{fig:hrs-unstable-gadget} and has a total of $1 + 3 =
      4$ residents and $1 + 2 = 3$ hospitals.}
    \label{fig:hrs-literal-gadget}
\end{figure}

\begin{figure}
    \centering
    \begin{picture}(68.80,123.20)
\put(24.00,123.20){\makebox(0,0){\fbox{$h_j^2$}}}
\put(56.00,123.20){\makebox(0,0){\fbox{$h_j^1$}}}
\put(0.00,3.20){\makebox(0,0){$r_j^4$}}
\put(0.00,35.20){\makebox(0,0){\fbox{$r_j^1$}}}
\put(0.00,67.20){\makebox(0,0){\fbox{$r_j^2$}}}
\put(0.00,99.20){\makebox(0,0){$r_j^3$}}
\put(24.00,3.20){\circle*{6.40}}
\put(24.00,8.00){\vector(0,1){54.40}}
\put(24.00,67.20){\circle*{6.40}}
\put(24.00,72.00){\vector(0,1){22.40}}
\put(24.00,99.20){\circle*{6.40}}
\put(28.80,99.20){\vector(1,0){22.40}}
\put(56.00,99.20){\circle*{6.40}}
\put(56.00,94.40){\vector(0,-1){54.40}}
\put(56.00,35.20){\circle*{6.40}}
\put(56.00,30.40){\vector(0,-1){22.40}}
\put(56.00,3.20){\circle*{6.40}}
\put(51.20,3.20){\vector(-1,0){22.40}}
\end{picture}
    \caption{The preference digraph for copy $j$ of the bistable
      gadget implicit in~\cite{McDermid2010} has four residents
      $r_j^1, \dots, r_j^4$ and two hospitals $h_j^1, h_j^2$.}
    \label{fig:hrs-bistable-gadget}
\end{figure}
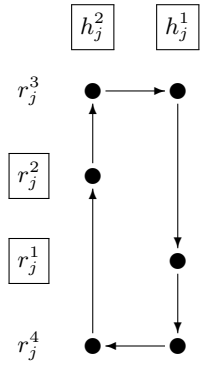

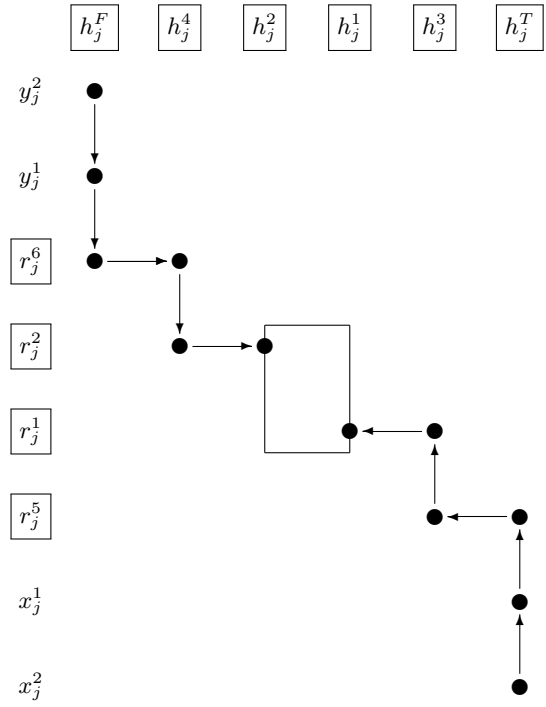
\begin{figure}
    \centering
    \begin{picture}(184.00,251.20)
\put(24.00,251.20){\makebox(0,0){\fbox{$h_j^F$}}}
\put(56.00,251.20){\makebox(0,0){\fbox{$h_j^4\vphantom{h_j^F}$}}}
\put(88.00,251.20){\makebox(0,0){\fbox{$h_j^2\vphantom{h_j^F}$}}}
\put(120.00,251.20){\makebox(0,0){\fbox{$h_j^1\vphantom{h_j^F}$}}}
\put(152.00,251.20){\makebox(0,0){\fbox{$h_j^3\vphantom{h_j^F}$}}}
\put(184.00,251.20){\makebox(0,0){\fbox{$h_j^T$}}}
\put(0.00,3.20){\makebox(0,0){$x_j^2$}}
\put(0.00,35.20){\makebox(0,0){$x_j^1$}}
\put(0.00,67.20){\makebox(0,0){\fbox{$r_j^5$}}}
\put(0.00,99.20){\makebox(0,0){\fbox{$r_j^1$}}}
\put(0.00,131.20){\makebox(0,0){\fbox{$r_j^2$}}}
\put(0.00,163.20){\makebox(0,0){\fbox{$r_j^6$}}}
\put(0.00,195.20){\makebox(0,0){$y_j^1$}}
\put(0.00,227.20){\makebox(0,0){$y_j^2$}}
\put(88.00,91.20){\line(1,0){32.00}}
\put(88.00,139.20){\line(1,0){32.00}}
\put(88.00,91.20){\line(0,1){48.00}}
\put(120.00,91.20){\line(0,1){48.00}}
\put(88.00,131.20){\circle*{6.40}}
\put(120.00,99.20){\circle*{6.40}}
\put(60.80,131.20){\vector(1,0){22.40}}
\put(56.00,131.20){\circle*{6.40}}
\put(56.00,158.40){\vector(0,-1){22.40}}
\put(56.00,163.20){\circle*{6.40}}
\put(28.80,163.20){\vector(1,0){22.40}}
\put(28.80,163.20){\vector(1,0){22.40}}
\put(24.00,163.20){\circle*{6.40}}
\put(24.00,190.40){\vector(0,-1){22.40}}
\put(24.00,195.20){\circle*{6.40}}
\put(24.00,222.40){\vector(0,-1){22.40}}
\put(24.00,227.20){\circle*{6.40}}
\put(147.20,99.20){\vector(-1,0){22.40}}
\put(152.00,99.20){\circle*{6.40}}
\put(152.00,72.00){\vector(0,1){22.40}}
\put(152.00,67.20){\circle*{6.40}}
\put(179.20,67.20){\vector(-1,0){22.40}}
\put(184.00,67.20){\circle*{6.40}}
\put(184.00,40.00){\vector(0,1){22.40}}
\put(184.00,35.20){\circle*{6.40}}
\put(184.00,8.00){\vector(0,1){22.40}}
\put(184.00,3.20){\circle*{6.40}}
\end{picture}
    \caption{The preference digraph for copy $j$ of the variable
      gadget implicit in~\cite{McDermid2010} includes an icon
      representing a copy of the bistable gadget of
      Figure~\ref{fig:hrs-bistable-gadget} and has a total of $6 + 4 =
      10$ residents and $4 + 2 = 6$ hospitals. This preference digraph
      can be streamlined as in Figure~\ref{fig:variable-gadget} to
      reduce the number of residents to $4+4 = 8$ and the number of
      hospitals to $2+2 = 4$.}
    \label{fig:hrs-variable-gadget}
\end{figure}
\fi

\end{document}